\documentclass[preprint,12pt]{elsarticle}

\usepackage{amssymb}
\usepackage{amsmath,amsfonts}
\usepackage{algorithmic}
\usepackage{array}
\usepackage[caption=false,font=normalsize,labelfont=sf,textfont=sf]{subfig}
\usepackage{textcomp}
\usepackage{stfloats}
\usepackage{url}
\usepackage{verbatim}
\usepackage{graphicx}
\def\BibTeX{{\rm B\kern-.05em{\sc i\kern-.025em b}\kern-.08em
    T\kern-.1667em\lower.7ex\hbox{E}\kern-.125emX}}
\usepackage{balance}

\usepackage{booktabs}                        
\usepackage{subcaption}   
\usepackage{color}
\usepackage{amsmath}
\usepackage{amssymb}
\usepackage{graphicx}
 \usepackage{dsfont}
\usepackage{amsfonts}
\usepackage{amsthm} 
\usepackage{stmaryrd}
\usepackage[all]{xy}
\usepackage{multirow}

\usepackage{qcircuit}
 
\def\>{\ensuremath{\rangle}}
\def\<{\ensuremath{\langle}}

\newcommand{\hs}{\mathcal{H}}

\newcommand {\tr} {{\mathit{tr}}}

\newtheorem{thm}{Theorem}[section]

\newtheorem{lem}{Lemma}[section]
\newtheorem{defn}{Definition}[section]
\newtheorem{prop}{Proposition}[section]
\newtheorem{exam}{Example}[section]

\journal{}

\begin{document}

\begin{frontmatter}

\title{Symbolic Specification and Reasoning for Quantum Data and Operations}

\author{Mingsheng Ying}

\affiliation{organization={Centre for Quantum Software and Information, University of Technology Sydney}, 
           addressline={15 Broadway}, 
           city={Ultimo},
           postcode={2007}, 
          state={NSW},
            country={Australia}}

\begin{abstract}
In quantum information and computation research, symbolic methods have been widely used for human specification and reasoning about quantum states and operations. At the same time, they are essential for ensuring the scalability and efficiency of automated reasoning and verification tools for quantum algorithms and programs. However, a formal theory for symbolic specification and reasoning about quantum data and operations is still lacking, which significantly limits the practical applicability of automated verification techniques in quantum computing. 

In this paper, we present a general logical framework, called Symbolic Operator Logic $\mathbf{SOL}$, which enables symbolic specification and reasoning about quantum data and operations. Within this framework, a classical first-order logical language is embedded into a language of formal operators used to specify quantum data and operations, including their recursive definitions. This embedding allows reasoning about their properties modulo a chosen theory of the underlying classical data (e.g., Boolean algebra or group theory), thereby leveraging existing automated verification tools developed for classical computing. It should be emphasised that this embedding of classical first-order logic into $\mathbf{SOL}$ is precisely what makes the symbolic method possible.

We envision that this framework can provide a conceptual foundation for the formal verification and automated theorem proving of quantum computation and information in proof assistants such as Lean, Coq, and related systems.
\end{abstract}

\begin{keyword}
Quantum computing\sep quantum information\sep formal specification\sep formal verification\sep symbolic techniques
\end{keyword}

\end{frontmatter}

\section{Introduction}

 A variety of formal verification tools for quantum algorithms and programs have been developed over the past decade. Notable examples include QWIRE \cite{Rand17}, SQIR \cite{Rand21}, and CoqQ \cite{Zhou23}, which are based on the Coq proof assistant; QHLProver \cite{Liu19}, qrhl-tool \cite{Unruh19}, and IMD \cite{Bord21}, which are built on Isabelle/HOL; QBRICKS \cite{Brick21}, which uses Why3 and SMT solvers; and tools based on ZX-calculus \cite{ZX20}. Since quantum states, operations, and predicates (used as pre/postconditions of quantum programs \cite{DP06}, as in quantum Hoare logic \cite{Ying11}) are mathematically formalized in terms of vectors and matrices -- for example, the density matrix of an $n$-qubit mixed state is a $2^n \times 2^n$ matrix, and a quantum gate on $n$ qubits is similarly represented -- the main bottleneck for practical and scalable application of these tools lies in the specification and manipulation of large matrices. To address this challenge, several approaches have been proposed, including tensor networks \cite{tensor}, decision diagrams \cite{Wille}, automata-based representations \cite{Chen}, abstract interpretation \cite{Yu21}, and separation logic \cite{QSL1, QSL2, QSL3}.

A particularly effective approach to improving the scalability of verification techniques in classical computing is the symbolic methods (e.g. symbolic model checking \cite{Clarke} and symbolic execution \cite{King}). In particular, in deductive verification of classical programs based on Hoare logic \cite{Apt09}, pre/postconditions are specified as first-order logical formulas, and verification conditions are proved employing formal logical reasoning. So, an important question for our purpose of scalable and efficient verification of quantum programs is how can we extend symbolic methods into quantum computing.   

\subsection{Motivating examples} Indeed, the idea of symbolic methods has already been widely employed in research on quantum information and computation. Quantum data (i.e. pure and mixed quantum states) and quantum operations -- including unitary transformations (or quantum gates) and quantum measurements (or observables) -- are typically modelled using the language of linear algebra, where states and operations are represented as vectors and matrices over complex numbers, in line with the formalism of quantum mechanics. In the study of quantum computation and quantum information (and also in quantum physics), however, these vectors and matrices often contain not just concrete complex numbers but also symbolic variables, functions, and expressions. Moreover, it is common to describe multi-qubit systems, where symbolic subscripts are used to indicate different qubits, as illustrated below:

\begin{exam}[Symbolic specification of quantum states]\label{ex-spec}
In the quantum computation and information literature, it is common to represent a quantum state in a form such as
 \begin{equation}\label{exam-1}\cos\frac{\theta}{2}|2k+1\rangle_{q[m+1]}|n-5\rangle_{q[3m-4]}+\sin\frac{\theta}{2}|2k-1\rangle_{q[m+1]}|n+3\rangle_{q[3m-4]}\end{equation} where the subscripts 
$q[m+1]$ and $q[3m-4]$ denote the $(m+1)$st and $(3m-4)$th subsystems, respectively, of a composite quantum system. In a sense, quantum state (\ref{exam-1}) is parameterised by classical variables $\theta, k, m, n$. 
\end{exam}

Furthermore, many results in quantum computation and information are derived through symbolic manipulation and reasoning. A simple example is the Z–Y decomposition for a single qubit (see Theorem 4.1 in \cite{NC00}):

\begin{exam}[Symbolic reasoning about quantum gates]\label{ex-reas} Let $I,X,Y, Z$ stand for the Pauli gates on a qubit and the rotation operators about the $x,y,z$ axes be defined as follows: 
\begin{align*}R_x(\theta)=\cos\frac{\theta}{2} I-i\sin\frac{\theta}{2}X,\qquad R_y(\theta)=\cos\frac{\theta}{2} I-i\sin\frac{\theta}{2}Y, \qquad R_z(\theta)=\cos\frac{\theta}{2} I-i\sin\frac{\theta}{2}Z
\end{align*} where real number $\theta$ is the rotation angle. Then the Z-Y decomposition for a gate $U=\left(a_{ij}\right)_{i,j=0}^1$ on a single qubit $q$ asserts:
\begin{equation}\label{ZY-eq}\begin{split}|a_{00}|^2+|a_{01}|^2=|a_{10}|^2+&|a_{11}|^2=1\wedge a_{00}a_{10}^\ast +a_{01}a_{11}^\ast=0\\ &\Rightarrow (\exists\theta,\theta_1,\theta_2,\theta_3)\left(\sum_{i,j=0}^1a_{ij}|i\rangle_q\langle j|=e^{i\theta}R_z(\theta_1)[q] R_y(\theta_2)[q] R_z(\theta_3)[q]\right)
\end{split}\end{equation} where $R_z(\theta_1)[q], R_y(\theta_2)[q]$ and $R_z(\theta_3)[q]$ mean that gate $R_z(\theta_1), R_y(\theta_2)$ and $R_z(\theta_3)$ act on qubit $q$, respectively. Note that quantum gates $U, R_z, R_y$ are parameterised by classical variables $a_{ij}, \theta, \theta_1, \theta_2, \theta_3$.    
\end{exam}

Clearly, the proof of the Z–Y decomposition theorem relies on a logical mechanism that allows one to infer the properties of quantum operations on the right-hand side of (\ref{ZY-eq}) from a theory of complex numbers satisfying the conditions on the left-hand side of (\ref{ZY-eq}). It is desired to verify (2) symbolically in a proof assistant such as Lean or Coq. 

Let us further explore how symbolic methods can aid in the verification of quantum programs by examining the quantum teleportation protocol.

\begin{exam}[Symbolic verification of quantum programs]\label{ex-veri} The four Bell states or EPR pairs
\begin{align*}&|\beta_{00}\rangle=\frac{|00\rangle+|11\rangle}{\sqrt{2}},\qquad |\beta_{01}\rangle=\frac{|01\rangle+|10\rangle}{\sqrt{2}},\\ 
&|\beta_{10}\rangle=\frac{|00\rangle-|11\rangle}{\sqrt{2}},\qquad |\beta_{11}\rangle=\frac{|01\rangle-|10\rangle}{\sqrt{2}}.
\end{align*} can be symbolically specified as $$|\beta_{xy}\rangle=\frac{|0,y\rangle+(-1)^x|1,\overline{y}\rangle}{\sqrt{2}}$$ using classical Boolean variables $x,y$ as parameters, where $\overline{y}$ denotes the negation of $y$. 

The teleportation of any state $|\psi\rangle=\alpha|0\rangle+\beta|1\rangle$ of a qubit $q$ using Bell state $|\beta_{xy}\rangle$ between qubits $q_a,q_b$ possessed by Alice and Bob, respectively can be written as the following program:
\begin{align*}\mathit{TEL}_{xy}:=\ &\mathit{CNOT}[q,q_a];H[q];m:=M[q];m_a:=M[q_a];\\ &\mathbf{if}\ y\oplus m_a=1\ \mathbf{then}\ X[q_b];\ \mathbf{if}\ x\oplus m=1\ \mathbf{then}\ Z[q_b];\ \mathbf{if}\ xm_a=1\ \mathbf{then}\ \mathit{Ph}
\end{align*} where $M$ stands for the measurement in the computational basis $|0\rangle,|1\rangle$, $\oplus$ for addition modulo $2$, and $\mathit{Ph}$ for global phase $-1$. \end{exam} 
 
The correctness of the program $\mathit{TEL}_{xy}$ for the case $(x, y) = (0, 0)$ is analysed in Section 1.3.7 of \cite{NC00}. To verify its correctness for the other three cases: $(x, y) = (0, 1), (1, 0), (1, 1)$, one could perform a similar analysis for each. However, a more elegant and efficient approach is to use symbolic techniques to reason about $\mathit{TEL}_{xy}$ with formal parameters $x$ and $y$.

\subsection{The aim of this paper} 

Symbolic specification and reasoning are already widely used in quantum information and computation research. But these applications are not conducted in a disciplined or formally structured way. In particular, a rigorous logical foundation for such reasoning is still lacking. This stands in sharp contrast to the situation in classical computing, where symbolic methods are underpinned by well-established logical frameworks. The absence of such foundations in the quantum setting poses a serious obstacle to the effective application of symbolic methods in the automated formal verification of quantum algorithms and programs.

The aim of this paper is to outline a general logical framework to support symbolic specification and reasoning about quantum data and operations as exemplified in Examples \ref{ex-spec} to \ref{ex-veri}. The basic design decisions for this framework are as follows: 
\begin{enumerate}\item We adopt \textbf{\textit{a single class of syntactic entities}}, called \textit{formal operators}, for the specification of both quantum data and operations. The idea behind this design is the simple observation that the vector representing an $n$-dimensional quantum state can be seen as an $n\times 1$ matrix. The benefit of this design is in fact similar to that of the Dirac notation in quantum physics, where the manipulations of quantum states, unitary transformations and observables can often be mixed together to achieve provided all of them are written in the Dirac notation. In particular, this enables us to present inference rules for symbolic reasoning about quantum data and operations in a unified manner. The separation between quantum data and operations will then be done using different \textit{signatures} (or \textit{types}) of formal operators. 

{\vskip 4pt}

The key idea behind our symbolic specification and reasoning is that formal operators are allowed to contain classical expressions and are therefore parameterised by classical variables. This is precisely why they are referred to as \textit{formal} operators.
It is indeed an extension of formal quantum states and formal quantum predicates introduced in \cite{Ying24, Ying24a}. 

\item We assume a classical first-order logical language $\mathcal{L}$. On top of it, we define \textit{\textbf{a syntax of formal operators}}, in which complex formal operators, denoted by $A, B, \dots$, are constructed from atomic operators using a set of constructors.
For example, we may take $|s\rangle_q$, $|s\rangle_q\langle t|$, etc. as atomic operators, where $s,t$ are classical expressions in $\mathcal{L}$ and $q$ is a quantum array variable. The constructors may include linear combination and the tensor product $\otimes$. Note that the expressions $s$ and $t$ can contain not only constants (e.g., complex numbers) but also classical variables. This reflects the idea of formal operators parameterised by classical variables as discussed in (1). Furthermore, the use of quantum arrays in this syntax provides a convenient mechanism for labelling different subsystems in a many-body quantum system.

{\vskip 4pt}

Introducing such a syntax often allows for much more compact representations of quantum data and operations than their matrix counterparts, especially when the matrices involved are sparse. This is analogous to the fact that first-order logical formulas can represent classical predicates more compactly than Boolean matrices. Furthermore, this syntax is also key to the symbolic manipulation, execution, and reasoning of quantum data and operations.

{\vskip 4pt}

Essentially, (1) and (2) provide a formal formulation and symbolic extension of the ideas underlying Dirac notation, which is widely used in quantum physics, computation, and information. The key idea here is that classical first-order expressions are embedded within Dirac notation, so that formal operators become parameterised by the involved classical variables (see Figure \ref{fig 0}).
\begin{figure}[h]\centering 
\begin{equation*}
\begin{array}{ccccc}\mbox{Symbolic Operator Logic}\ \mathbf{SOL}\\ \\ \Uparrow\ \mbox{embedding}\\ \\ \mbox{First-Order Logic}\ \mathcal{L}
\end{array}
\end{equation*} \caption{Two-Layer Logical Framework.}\label{fig 0}
\end{figure}
\item Our next step is to construct a first-order logic over the formal operator formulas defined in (2), which we call \textbf{\textit{Symbolic Operator Logic}} ($\mathbf{SOL}$). In this logic, formal operators are treated as expressions (i.e., terms). We then introduce various basic properties of, or relations between, formal operators --- for example, $A \otimes B = C$ --- and use them as atomic propositions to form logical formulas specifying more complex properties of quantum data and operations.

{\vskip 4pt} 

Furthermore, we define entailment of the form
\begin{equation}\label{O-entail}
\Sigma \models \mathcal{A},
\end{equation}
where $\Sigma$ is a theory in the logic $\mathcal{L}$, and $\mathcal{A}$ is a logical formula in the symbolic operator logic $\mathbf{SOL}$. This enables (conditional) reasoning --- such as equational reasoning using rewriting techniques --- about properties of quantum data and operations represented as logical formulas (e.g., $\mathcal{A}$) in $\mathbf{SOL}$, modulo a chosen theory (e.g., $\Sigma$) describing the involved classical data and operations in logic $\mathcal{L}$. For instance, we can reason as follows:
\begin{equation}\label{eq-address}\begin{split}
\left(k=\frac{3}{2}m-2\right) &\wedge \left(n=\frac{5}{7}l-1\right) \\ 
&\models \mbox{CNOT}[q[2k+3],q[5l-2]];\mbox{CNOT}[q[3m-1],q[7n+5]] \equiv \mathbf{skip}.
\end{split}\end{equation}
Another example is that the theory of Boolean arithmetic for variables $x$ and $y$ is required in the verification of the program $\mathit{TEL}_{xy}$ in Example~\ref{ex-veri}. The entailment (\ref{O-entail}) can be further generalised to \begin{equation}\label{O-entails}\Sigma,\Gamma\models\mathcal{A}\end{equation} where $\Gamma$ is a set of logical formulas in $\mathbf{SOL}$, indicating that the theories $\Sigma$ 
 (classical) and $\Gamma$ (quantum) jointly entail a property $\mathcal{A}$ of quantum entities.

{\vskip 4pt} 

The introduction of entailments~(\ref{O-entail}) and (\ref{O-entails}) and the reasoning they enable provides a convenient mechanism for leveraging verification tools developed in classical computing to reason about quantum data and operations.
\end{enumerate}

\subsection{Organisation of the paper}

For the reader’s convenience and to fix the notation, the classical first-order logic $\mathcal{L}$ upon which our symbolic operator logic $\mathbf{SOL}$ is defined is briefly reviewed in Section \ref{sec-classic}. The basic ingredients --- quantum types, arrays and variables --- of $\mathbf{SOL}$ are introduced in Section \ref{sec-arrays}. The syntax and semantics of formal operators as terms in $\mathbf{SOL}$ are defined in Section \ref{sec-operators}. Various basic functions and predicates over formal operators are introduced in Section \ref{sec-fp}. The syntax and semantics of logical formulas in $\mathbf{SOL}$; in particular, the notion of entailment between logical formulas in $\mathbf{SOL}$ conditional on a classical first-order theory expressed in the logic $\mathcal{L}$, are defined in Section \ref{sec-sol}. Examples illustrating symbolic reasoning about properties of quantum data and operations in $\mathbf{SOL}$ are presented in Section \ref{sec-exam1}, and several examples of recursive definitions of quantum data and operations in $\mathbf{SOL}$ are provided in Section \ref{sec-exam2}. Potential further applications of $\mathbf{SOL}$ are discussed in Section \ref{sec-exam3}. A brief conclusion is drawn in Section \ref{sec-concl}.

\section{Classical First-Order Language}\label{sec-classic}

Our symbolic operator logic $\mathbf{SOL}$ is built upon a classical first-order logic $\mathcal{L}$. For the reader’s convenience, let us first briefly introduce the syntax and semantics of $\mathcal{L}$ in this section. We adopt a typed first-order language $\mathcal{L}$, following the presentation in the textbook \cite{Apt09} for classical program verification.
 
\subsection{Types}\label{c-types} Intuitively, a type is a name or notation for an intended set of values. In the logic $\mathcal{L}$, we assume a set of basic types. For example, \begin{itemize}\item $\mathbf{Bool}$ denotes the set $\{\mathbf{false}, \textbf{true}\}$; \item $\mathbf{Int}$ denotes the set of all integers $...,-2,-1,0,1,2,...$; \item In particular, we need the type $\mathbf{C}$ of complex numbers since the language defined in this section will be used later in the description of quantum states and operations.\end{itemize}

For any basic types $T_1,...,T_n$ $(n\geq 1)$ and $T$, we introduce a higher type $T_1\times ...\times T_n\rightarrow T$. Intuitively, it denotes the set of all functions from the Cartesian product of the sets denoted by $T_1, ..., T_n$  to the set denoted by $T$. Here, $T_1,...,T_n$ are called the argument types, and $T$ the value type.  

\subsection{Syntax}

The \textbf{alphabet} of our first-order language $\mathcal{L}$ consists of:\begin{itemize}\item a set $\mathit{Var}$ of variables;
\item a set of constants;
\item logical connectives $\neg$ (negation), $\wedge$ (conjunction), $\vee$ (disjunction), $\rightarrow$ (implication);
\item universal quantifiers $\forall$, existential quantifier $\exists$.
\end{itemize}

Each variable $x$ is equipped with a type $T_x$. If $T_x$ is a basic type, then $x$ is called a simple variable; if $T_x$ is a higher type, then $x$ is called an array variable. Similarly, each constant is also equipped with a type. For example, integer constants $..., -2,-1,0,1,2,...$ have the basic type $\mathbf{Int}$, boolean constants $\mathbf{true},\mathbf{false}$ have the basic type $\mathbf{Bool}$, and each complex number has the basic type $\mathbf{C}$. Furthermore, a function symbol can be seen as a constant of a higher type $T_1\times ...\times T_n\rightarrow T$. Similarly, a predicate symbol can be seen as a constant of a higher type $T_1\times ...\times T_n\rightarrow\mathbf{Bool}$.

{\vskip 4pt}

\textbf{Expressions} (i.e. terms) are formed from variables and constants by applying function symbols.  
\begin{defn}\label{def-c-expression} Expressions in $\mathcal{L}$ are inductively defined as follows: \begin{enumerate}\item A simple variable and a constant of basic type $T$ are both expressions of type $T$; 
\item If $s_1,...,s_n$ are expressions of types $T_1,...,T_n$, respectively, and $\mathit{op}$ is a constant of type $T_1\times ...\times T_n\rightarrow T$, then \begin{equation}\label{op-syn}\mathit{op}(s_1,...,s_n)\end{equation} is an expression of type $T$;
\item If $s_1,...,s_n$ are expressions of types $T_1,...,T_n$, respectively, and $a$ is an array variable of type $T_1\times ...\times T_n\rightarrow T$, then \begin{equation}\label{sub-syn}a[s_1,...,s_n]\end{equation} is an expression of type $T$, often called a subscripted variable; 
\item If $b$ is a boolean constant and both $s_0,s_1$ are expressions of type $T$, then \begin{equation}\label{if-syn}\mathbf{if}\ b\ \mathbf{then}\ s_1\ \mathbf{else}\ s_0\ \mathbf{fi}\end{equation} is an expression of type $T$.
\end{enumerate}\end{defn}

It is clear from the definition that all expressions have a basic type. We write $\mathit{var}(s)$ for the set of variables in expression $s$.  

\textbf{Logical formulas} are constructed from atomic propositions (i.e. predicate symbols applied to expressions, i.e. terms) using logical connectives and quantifiers. 

\begin{defn} First-order logical formulas in $\mathcal{L}$ are inductively defined as follows:\begin{enumerate}\item If $p$ is a predicate symbol, i.e. a constant of type $T_1\times\cdots\times T_n\rightarrow\mathbf{Bool}$, and $s_1,...,s_n$ are expressions of type $T_1,...,T_n$, respectively, then $p(s_1,...,s_n)$ is a logical formula, called atomic proposition; \item If $\varphi,\psi$ are logical formulas, so are $\neg\varphi, \varphi\wedge\psi,\varphi\vee\psi$ and $\varphi\rightarrow\psi$;
\item If $\varphi$ is a logical formula and $x$ a variable, then $(\forall x)\varphi$ and $(\exists x)\varphi$ are logical formulas. 
\end{enumerate}
\end{defn}

We write $\mathit{free}(\varphi)$ for the set of free variables of $\varphi$, that is, the variables not bound by the quantifiers. The set of logical formulas in $\mathcal{L}$ is denoted by $\mathit{Wff}_\mathcal{L}$. 

{\vskip 4pt}

\textbf{Substitution} is a basic operation that replaces variables in an expression or logical formula with specific terms or expressions. It enables us to apply general logical rules to particular cases.

\begin{defn}\label{def-term-sub}\label{def-c-sub} Let $x$ be a simple or subscripted variable of type $T$ and let $t$ an expression of the same type. Then:\begin{enumerate}\item For any expression $s$, expression $s[t/x]$ is obtained substitution of $t$ for $x$ in $s$. More precisely, it is inductively defined as follows:\begin{enumerate}\item If $s=y$ is a simple variable with $y\neq x$, or $s=c$ is a constant, then $s[t/x]=s$; if $s=x$, then then $s[t/x]=t$.
\item If $s=\mathit{op}(s_1,...,s_n)$, then $s[t/x]=\mathit{op}(s_1[t/x],...,s_n[t/x])$. 
\item Let $s=a[s_1,...,s_n]$ for an array $a$. If $x$ is a simple variable or $x=a^\prime[s_1^\prime,...,s_m^\prime]$ with $a^\prime\neq a$, then $s[t/x]=a[s_1[t/x],...,s_n[t/x]]$; if $x=a[s_1^\prime,...,s_m^\prime]$, then 
$$s[t/x]=\mathbf{if}\ \bigwedge_{i=1}^ns_i[t/x]=s_i^\prime\ \mathbf{then}\ t\ \mathbf{else}\  a[s_1[t/x],...,s_n[t/x]].$$  
\end{enumerate}
\item For any logical formula $\varphi$, substitution $\varphi[t/x]$ is inductively defined as follows:\begin{enumerate}\item If $\varphi=p(s_1,...,s_n)$ is an atomic proposition, then $\varphi[t/x]=p(s_1[t/x],...,s_n[t/x])$.
\item If $\varphi=\neg\psi$, then $\varphi[t/x]=\neg\psi[t/x]$; if $\ast\in\{\wedge,\vee,\rightarrow\}$ and $\varphi=\varphi_1\ast\varphi_2$, then $\varphi[t/x]=\varphi_1[t/x]\ast\varphi_2[t/x]$. 
\item Let $Q\in\{\forall,\exists\}$ and $\varphi=(Qy)\psi$. If $y=x$ then $\varphi[t/x]=\varphi$; if $y\neq x$ then $\varphi[t/x]=(Qy^\prime)\psi[y^\prime/y][t/x]$, where $y^\prime\notin\mathit{free}(Q)\cup\mathit{var}(t)\cup\{x\}$.
\end{enumerate}
\end{enumerate}
\end{defn}

The above notion of substitution can be generalised to simultaneous substitution $s[\overline{t}/\overline{x}]$ of a string $\overline{t}$ of expressions for a string $\overline{x}$ of distinct variables.   
 
\subsection{Semantics}

As an interpretation of our logical language $\mathcal{L}$, we introduce the notion of first-order structure. It is defined as a pair $\mathcal{S}=(\mathcal{D},\mathbb{I})$, where:
\begin{itemize}\item $\mathcal{D}$ is a family $\left\{\mathcal{D}_T\right\}$, where for each basic type $T$, $\mathcal{D}_T$ is a nonempty set, called the domain of type $T$; in particular, $\mathcal{D}_{\mathbf{Int}}=\{...,-2,-1,0,1,2,...\}$, $\mathcal{D}_{\mathbf{Bool}}=\{\mathbf{true},\mathbf{false}\}$, and $\mathcal{D}_\mathbf{C}$ is the set of complex numbers. Then for a higher type, we define:\begin{equation}\label{highD}\mathcal{D}_{T_1\times ...\times T_n\rightarrow T}=\mathcal{D}_{T_1}\times ...\times \mathcal{D}_{T_n}\rightarrow \mathcal{D}_{T}\end{equation} (the set of all mappings from the Cartesian product of $\mathcal{D}_{T_1},..., \mathcal{D}_{T_n}$ into $\mathcal{D}_{T}$);
\item $\mathbb{I}$ is an interpretation; that is, each constant $c$ of a basic or higher type $T$ is interpreted as an element $c^{\mathbb{I}}$ in $\mathcal{D}_T$.
In particular: \begin{itemize}\item a function symbol $F$ of type $T_1\times ...\times T_n\rightarrow T$ is interpreted as a mapping $F^\mathbb{I}:\mathcal{D}_{T_1}\times ...\times\mathcal{D}_{T_n}\rightarrow\mathcal{D}_T$; and \item a predicate symbol $P$ of type $T_1\times ...\times T_n\rightarrow \mathbf{Bool}$ is interpreted as a mapping $P^\mathbb{I}:\mathcal{D}_{T_1}\times ...\times\mathcal{D}_{T_n}\rightarrow\{\mathbf{true},\mathbf{false}\}$. \end{itemize}
\end{itemize}

Given a structure $\mathcal{S}$ with domains $\mathcal{D}$. Then a state in $\mathcal{S}$ is a mapping $\sigma$ from each variable $x$ of a basic or higher type $T$ to an element $\sigma(x)$ in $\mathcal{D}_T$. The set of all states in $\mathcal{S}$ is denoted by $\mathit{Stat}_\mathcal{S}$. We often simply write $\mathit{Stat}$ for $\mathit{Stat}_\mathcal{S}$ if the structure $\mathcal{S}$ is fixed or can be recognised from the context.    

Let $\sigma$ be a state, $x$ a variable and $d$ a value in $\mathcal{D}_{T}$, where $T$ is the type of $x$. Then the updated state $\sigma[x:=d]$ is defined as follows: for any variable $x^\prime$, \begin{equation}\label{eq-updated}\sigma[x:=d](x^\prime)=\begin{cases}d\ &\mbox{if}\ x^\prime=x; \\ \sigma(x)\ &\mbox{if}\ x^\prime\neq x.
\end{cases}\end{equation}

\begin{defn}\label{def-term-value} The value $\sigma(s)$ of an expression $s$ in a state $\sigma\in\Sigma$ is inductively defined as follows:\begin{enumerate}\item If $s$ is a simple variable, then $\sigma(s)$ is given by $\sigma$, and if $s$ is a constant of a basic type, then $\sigma(s)=s^{\mathbb{I}}$; 
\item If $s=\mathit{op}(s_1,...,s_n)$, then $\sigma(s)=\mathit{op}^{\mathbb{I}}(\sigma(s_1),...,\sigma(s_n))$;
\item If $s=a[s_1,...,s_n]$, then $\sigma(s)=\sigma(a)(\sigma(s_1),...,\sigma(s_n));$
\item If $s=\mathbf{if}\ b\ \mathbf{then}\ s_1\ \mathbf{else}\ s_0\ \mathbf{fi}$, then $$\sigma(s)=\begin{cases}\sigma(s_1)\ &\mbox{if}\ \sigma(b)=\mathbf{true},\\ \sigma(s_0)\   &\mbox{otherwise}.\end{cases}$$
\end{enumerate}
\end{defn}

\begin{defn} The satisfaction relation $\sigma\models\varphi$, meaning that a state $\sigma\in\Sigma$ satisfies a logical formula $\varphi$, is inductively defined as follows:\begin{enumerate}\item $\sigma\models p(s_1,...,s_n)$ iff $p^{\mathbb{I}}(\sigma(s_1),...,\sigma(s_n))=\mathbf{true}$; \item $\sigma\models\neg\varphi$ iff not $\sigma\models\varphi$, $\sigma\models\varphi\wedge\psi$ iff $\sigma\models\varphi$ and $\sigma\models\psi$, $\sigma\models\varphi\vee\psi$ iff $\sigma\models\varphi$ or $\sigma\models\psi$, $\sigma\models\varphi\rightarrow\psi$ iff $\sigma\models\varphi$ implies $\sigma\models\psi$; 
\item If $x$ is a variable of type $T$, then:\begin{itemize}\item $\sigma\models(\forall x)\varphi$ iff $\sigma[x:=d]\models\varphi$ for any $d\in\mathcal{D}_T$, 
\item $\sigma\models(\exists x)\varphi$ iff $\sigma[x:=d]\models\varphi$ for some $d\in\mathcal{D}_T$. 
\end{itemize} 
\end{enumerate}
\end{defn}

The following lemma establishes a connection between substitution and state update. 

\begin{lem}[Substitution]\label{lem-classical-sub}For any variable $x$, expressions $s,t$ and logical formula $\varphi$, and for any state $\sigma$, we have:\begin{enumerate}\item $\sigma(s[\overline{t}/\overline{x}])=\sigma[\overline{x}:=\sigma(\overline{t})](s)$; \item $\sigma\models\varphi[\overline{t}/\overline{x}]$ iff $\sigma[\overline{x}:=\sigma(\overline{t})]\models\varphi$.  
\end{enumerate}
\end{lem}

We are not going to further review the proof system of classical first-order logic $\mathcal{L}$. To conclude this section, we would like to point out that the materials reviewed in this section will serve a twofold purpose:
\begin{itemize}
\item Our symbolic operator logic $\mathbf{SOL}$ will be built upon logic $\mathcal{L}$; in other words, $\mathcal{L}$ is embedded in $\mathbf{SOL}$.
\item Essentially, $\mathbf{SOL}$ is also a first-order logic (though the individuals in it are operators). Therefore, $\mathbf{SOL}$ will be developed following the pattern of $\mathcal{L}$.
\end{itemize}

\section{Quantum Types, Arrays and Variables}\label{sec-arrays}

From now on, we fix the first-order logical language $\mathcal{L}$ and a first-order structure as an interpretation of $\mathcal{L}$. Our symbolic operator logic $\mathbf{SOL}$ for specification and reasoning about quantum data and operations is built upon them. In this section, we start to introduce some basic ingredients of $\mathbf{SOL}$, namely quantum types, arrays and variables. These concepts were first introduced in \cite{Ying24a, Ying24b}.   

\subsection{Quantum types} 

First of all, we introduce a way of quantising some (\textit{but may not all}) of classical types:
\begin{itemize}\item For a basic type $T$, we assume a corresponding quantum type, also denoted $T$ by a slight abuse of notation. The quantum domain of type $T$ is then the Hilbert space spanned by the classical domain of $T$:
\begin{equation}\label{D2H}\hs_T=\mbox{span}\left\{|d\rangle: d\in\mathcal{D}_T\right\}.\end{equation} That is, the classical domain $\mathcal{D}_T$ is an orthonormal basis of $\hs_T$, called the standard basis. For example, the quantum domain $\hs_{\mathbf{Bool}}$ of type $\mathbf{Bool}$ is the $2$-dimensional Hilbert space, i.e. the state space of a qubit, and the quantum domain of type $\mathbf{Int}$ is $$\hs_{\mathbf{Int}}=\left\{\sum_{i=-\infty}^\infty \alpha_i|i\rangle:\sum_{i=-\infty}^\infty|\alpha_n|^2<\infty\right\}.$$

\item For a higher type $T_1\times ...\times T_n\rightarrow T$, we only quantise its value type $T$, but leave its argument types $T_1,...,T_n$ being classical. Thus, the quantum domain of $T_1\times ...\times T_n\rightarrow T$ is
\begin{equation}\label{highH}\hs_{T_1\times...\times T_n\rightarrow T}=\hs_T^{\otimes (\mathcal{D}_{T_1}\times ...\times \mathcal{D}_{T_n})}=\bigotimes_{d_1\in \mathcal{D}_{T_1},...,d_n\in \mathcal{D}_{T_n}}\hs_T\end{equation} i.e., the tensor product of multiple copies of $\hs_T$ indexed by values $d_1,...,d_n$ of type $T_1,...,T_n$, respectively. Essentially, (\ref{highH}) is obtained through replacing $\mathcal{D}_T$ by $\hs_T$ in (\ref{highD}) (but leaving $\mathcal{D}_{T_1},...,\mathcal{D}_{T_n}$ unchanged). 
\end{itemize}

\subsection{Quantum variables}\label{sec-qvar} 

Corresponding to the variables used in the classical language $\mathcal{L}$, we introduce two sorts of quantum variables:\begin{itemize}\item simple quantum variables, of a basic type, say $T$; \item array quantum variables, of a higher quantum type, say $T_1\times ...\times T_n\rightarrow T$.  
\end{itemize}

It is worth noting that the types of classical variables and quantum variables are the same, but their interpretations are different; more precisely, the value of a classical variable $x$ with type $T$ is an element in $\mathcal{D}_T$, whereas the value of a quantum variable $q$ with the same type $T$ is an vector in the Hilbert space $\hs_T$. Thus, the state space of $q$ is defined as $\hs_q=\hs_T$. 

The notion of subscripted variables in Definition \ref{def-c-expression}(3) can be straightforwardly generalised to the quantum setting:   

\begin{defn}Let $q$ be an array quantum variable of type $T_1\times ...\times T_n\rightarrow T$, and for each $1\leq i\leq n$, let $s_i$ be a classical expression of type $T_i$. Then \begin{equation}\label{qsub-syn}q[s_1,...,s_n]\end{equation} is called a subscripted quantum variable of type $T$.\end{defn}

Subscripted quantum variable (\ref{qsub-syn}) is defined in a way similar to subscripted classical variable (\ref{sub-syn}). Intuitively, array variable $q$ denotes a quantum system composed of subsystems indexed by $(d_1,...,d_n)\in \mathcal{D}_{T_1}\times ...\times \mathcal{D}_{T_n}$. Thus, whenever expression $s_i$ is evaluated to a value $d_i\in \mathcal{D}_{T_i}$ for each $i$, then $q[s_1,...,s_n]$ indicates the subsystem of index tuple $(d_1,...,d_n)$. For simplification of the presentation, a simple variable $q$ is often viewed as a subscripted variable $q[s_1,...,s_n]$ for $n=0$.  

\begin{exam}Let $q$ be a qubit array of type $\mathbf{Int}\times\mathbf{Int}\rightarrow\mathbf{Bool}$. Then $q[2x+y,7-3y]$ is a subscripted qubit variable. If $x=5$ and $y=-1$ in a classical state $\sigma$, then it stands for the qubit $q[9,10]$; that is, $\sigma(q[2x+y,7-3y])=q[9,10]$ (see Definition \ref{def-ground} below).\end{exam}

Let $\overline{q}=q_1,...,q_n$, where $q_i$ is a simple or subscripted quantum variable of type $T_i$ for $1\leq i\leq n$. Then the dimension of $\overline{q}$ is defined as 
$\dim\overline{q}=\prod_{i=1}^n|T_i|,$ where $|T_i|$ stands for the size of (classical) type $T_i$. Thus, the state space of $\overline{q}$ is the Hilbert space $\hs_{\overline{q}}=\bigotimes_{i=1}^n\hs_{q_i}$ of dimension $\dim\overline{q}$. 

The no-cloning theorem for quantum information prohibits that a quantum operation is performed on two identical qubits. To reflect this requirement in the syntax, we introduce:

\begin{defn}Let $\overline{q} = q_1[s_{11},...,s_{1n_1}],...,q_k[s_{k1},...,s_{kn_k}]$ be a sequence of simple or subscripted quantum variables. Then the distinctness for $\overline{q}$ is defined as the following first-order logical formula:
\begin{equation}\label{distinctness}\mathit{Dist}(\overline{q})=\bigwedge_{i,j\ \mbox{with}\ q_i=q_j}(\exists l\leq n_i)\left(s_{il}\neq s_{jl}\right).\end{equation}
\end{defn}

This distinctness will be used frequently throughout this paper to ensure the well-definedness of concepts involving multiple qubits. For a given classical structure $\mathcal{S}$ as an interpretation of the first-order logic $\mathcal{L}$ and a state $\sigma\in\mathit{Stat}_\mathcal{S}$ the satisfaction relation: $$\sigma\models\mathit{Dist}(q_1[t_{11},...,s_{1n_1}],...,q_k[s_{k1},...,s_{kn_k}])$$ means that in the classical state $\sigma$, $q_1[s_{11},...,s_{1n_1}],$ $...,q_k[s_{k1},...,s_{kn_k}]$ denotes $k$ distinct quantum systems. 

By a \textit{ground quantum variable} we mean a simple quantum variable or $q[d_1,...d_n]$, where $q$ is an array quantum variable of type $T_1\times ...\times T_n$ and $d_i\in\mathcal{D}_{T_i}$ for $1\leq i\leq n$. Then we can introduce: 

\begin{defn}\label{def-ground}Let $\sigma$ be a classical state in the given first-order structure $\mathcal{S}$. Then:\begin{itemize}\item for a simple quantum variable $q$, we define $\sigma(q)=q$;\item for a subscripted quantum variable, we define: \begin{equation}\label{sig-sub}\sigma(q[t_1,...,t_n])=q[\sigma(t_1),...,\sigma(t_n)]\end{equation} where $\sigma(t_i)$ stands for the value of expression $t_i$ in state $\sigma$ (see Definition \ref{def-term-value}).\end{itemize} Furthermore, 
for a string $\overline{q}=q_1,...,q_n$ of simple or subscripted quantum variables, we set: \begin{equation}\label{sig-string}\sigma(\overline{q})=\sigma(q_1),...,\sigma(q_n).\end{equation}\end{defn}
It should be noted that (\ref{sig-sub}) is a ground variable and (\ref{sig-string}) is a string of ground variables. 

\section{Formal Operators}\label{sec-operators}

Now we start to develop our symbolic operator logic $\mathbf{SOL}$. As the first step, in this section, we define the syntax and semantics of formal operators. As said before, operators are treated as individuals in logic $\mathcal{L}$. Therefore, formal operators defined in this section are in fact the expressions (i.e., terms of individuals) in $\mathbf{SOL}$. 

\subsection{Operator Variables and Constants}\label{op-var}

In the alphabet of $\mathbf{SOL}$, we assume:\begin{enumerate}\item a set of operator variables, ranged over by meta-symbols $X,X_1,X_2,...$; and \item a set of basic operator symbols, called operator constants, ranged over by meta-symbols $O, O_1, O_2,...$.\end{enumerate} 

Each operator variable $X$ is equipped with a quantum type of the form: \begin{equation}\label{op-type-0}T_1\times ...\times T_m\rightarrow T_1^\prime\times ...\times T_n^\prime\end{equation} where $T_1,...,T_m, T_1^\prime, ...,T_n^\prime$ are all basic types. Intuitively, a value of $X$ is a linear operator from Hilbert space  $\hs_{T_1\times ...\times T_m}$ to $\hs_{T^\prime_1\times ...\times T^\prime_n}$, where the quantum domain of a product type is defined by
\begin{equation}\label{op-type}\hs_{T_1\times ...\times T_m}=\bigotimes_{i=1}^m\hs_{T_i}.\end{equation}
It should be noticed that $m=0$ or $n=0$ in type (\ref{op-type}) is allowed. Let us write $[]$ for the empty type, i.e. $T_1\times ...\times T_m$ with $m=0$. Then:\begin{itemize}\item an operator variable $X$ of type $[]\rightarrow T_1^\prime\times ...T_n^\prime$ denotes a bra-vector $\langle\varphi|$ in Hilbert space $\hs_{T_1^\prime\times...\times T_n^\prime}$;
\item an operator variable $X$ of type $T_1\times ...T_m\rightarrow []$ denotes a ket-vector $|\varphi\rangle$ in Hilbert space $\hs_{T_1\times...\times T_m}$;
\item an operator variable $X$ of type $[] \rightarrow []$ denotes a complex number $c$. 
\end{itemize}
Whenever $T_1\times ...\times T_m=T_1^\prime\times ...\times T_n^\prime$, then we simply say that an operator variable $X$ of type (\ref{op-type-0}) has type $T_1\times ...\times T_m$. 

Similarly, each operator constant $O$ is equipped with a quantum type like (\ref{op-type-0}). Additionally, it is equipped a type for its classical parameters, say $T_1^{\prime\prime}\times ...\times T_k^{\prime\prime}$. This means that $O$ can carry $k$ classical (formal) parameters $x_1,...,x_k$ such that for each $1\leq i\leq k$, $x_i$ is a classical variable of type $T_i^{\prime\prime}.$

\begin{exam}\begin{enumerate}\item As operator constants, Hadamard gate $H$ and Pauli gates $I, X, Y, Z$ has quantum type $\mathbf{Bool}$, and the controlled-not gate $\mbox{CNOT}$ has quantum type $\mathbf{Bool}\times\mathbf{Bool}$. Since they do not carry any parameter, their parameter type is empty $[]$. 
\item The rotation operators $R_x,R_y$ and $R_z$ in Example \ref{ex-reas} have quantum type $\mathbf{Bool}$ and parameter type $\mathbf{R}$ (real numbers) for parameter $\theta$. 
\item The quantum state (\ref{exam-1}) in Example \ref{ex-spec} can be seen as an operator constant with quantum type $\mathbf{Int}\times\mathbf{Int}$ and parameter type $\mathbf{R}\times\mathbf{Int}\times\mathbf{Int}$ because it has three parameters $\theta,k$ and $m$. 
\end{enumerate}
\end{exam}

An operator variable, an operator constant or more generally a formal operator will be used to describe a transformation from a sequence of quantum variables (denoting a composite quantum system) to another sequence of quantum variables. To specify explicitly the domains and codomains of operators, we introduce:   

\begin{defn}[Signatures]\begin{enumerate}\item A signature is defined as a formula of the form \begin{equation}\label{eq-sign}\overline{q}\rightarrow\overline{q^\prime}\end{equation} where $\overline{q}$ and $\overline{q^\prime}$ are strings of simple or subscripted quantum variables (and can be empty string $\epsilon$).
\item If $\overline{q}$ and $\overline{q^\prime}$ are strings of ground quantum variables, then (\ref{eq-sign}) is called a ground signature.
\end{enumerate}
\end{defn}

Let $\overline{q}=q_1,...,q_m$ and $\overline{q^\prime}=q^\prime_1,...,q^\prime_n$. For each $i$, let $q_i$ have the type of $T_i$ and $q_i^\prime$ have type $T_i^\prime$. Then the type of signature (\ref{eq-sign}) is defined to be $T_1\times ...\times T_m\rightarrow T_1^\prime ...\times T_n^\prime$. In particular, let $\epsilon$ stand for the empty string of quantum variables, then signature $\overline{q}\rightarrow\epsilon$ has type $T_1\times ...\times T_m\rightarrow []$, signature $\epsilon\rightarrow\overline{q^\prime}$ has type $[]\rightarrow T_1^\prime\times...\times T_n^\prime$, and signature $\epsilon\rightarrow\epsilon$ has type $[]\rightarrow []$. As we will see shortly, when an operator symbol $O$ applies to a signature $\overline{q}\rightarrow\overline{q^\prime}$, their types must be matched. 

\subsection{Syntax of Formal Operators}\label{sec-fop-s}

Now we are able to present the syntax of formal operators. The alphabet of our language for formal operators includes that of the classical first-order logical language $\mathcal{L}$ in Section \ref{sec-classic} together with: \begin{itemize}\item quantum variables given in Section \ref{sec-arrays}; \item basic operator symbols $O, O_1, O_2,...$, called operator constant symbols, each of which is assigned a quantum type and a parameter type (as discussed in Subsection \ref{op-var}); \item operator variables $X, X_1,X_2,...$, each of which is assigned a quantum type;  
and \item operation symbols $\cdot$ (scalar multiplication), $^\dag$ (adjoint), $+$ (addition), $\cdot$ (multiplication) and $\otimes$ (tensor product). \end{itemize}

A common way to specify quantum data and operations is to represent them as matrices or operators. However, when the dimension of the Hilbert space under consideration is too large, it becomes impractical to specify them explicitly. Therefore, we need to introduce:

\begin{defn}[Syntax]\label{def-fop} Formal Operators (ranged over by $A,A_1,A_2,...$) are defined by the following syntax:
\begin{equation}\begin{split} A::= \ c \ &|\ |s\rangle_q \ | \ \langle s|_q\ |\ X[\overline{q}\rightarrow\overline{q^\prime}] \ |\ O(t_1,...,t_k)[\overline{q}\rightarrow\overline{q^\prime}]\\ & |\ cA\ |\ A^\dag\ |\ A_1+A_2\ | \ A_1A_2\ |\ A_1\otimes A_2
\end{split}\end{equation} More precisely, a formal operator $A$ and its signature $\mathit{Sign}(A)$ are defined as follows:\begin{enumerate}\item a classical expression $c$ of type $\mathbf{C}$ (complex numbers) is a formal operator, and $\mathit{Sign}(c)=[]\rightarrow []$;
\item if $s$ is a classical expression of the same type as simple or subscripted quantum variable $q$, then both $|s\rangle_q$ and $\langle s|_q$ are formal operators, and $\mathit{Sign}(|s\rangle_q)=q\rightarrow []$, $\mathit{Sign}(\langle s|_q)=[]\rightarrow q$; 
\item if $X$ is a operator variable of the same type as signature $\overline{q}\rightarrow\overline{q^\prime}$, then $X[\overline{q}\rightarrow\overline{q^\prime}]$ is a formal operator, and $\mathit{Sign}(X[\overline{q}\rightarrow\overline{q^\prime}])=\overline{q}\rightarrow\overline{q^\prime}$. In particular, whenever $\overline{q}=\overline{q^\prime}$, then $X[\overline{q}\rightarrow\overline{q^\prime}]$ can be abbreviated to $X[\overline{q}]$;
\item if $\overline{t}$ is a string of classical expressions with types $T_1,...,T_k$, respectively, and $O$ is an operator constant of the same quantum type as signature $\overline{q}\rightarrow\overline{q^\prime}$ and assigned parameter type $T_1\times ...\times T_k$, then $O(\overline{t})[\overline{q}\rightarrow\overline{q^\prime}]$ is a formal operator, and $\mathit{Sign}(O(\overline{t})[\overline{q}\rightarrow\overline{q^\prime}])=\overline{q}\rightarrow\overline{q^\prime}$. In particular, whenever $\overline{q}=\overline{q^\prime}$, then $O(\overline{t})[\overline{q}\rightarrow\overline{q^\prime}]$ can be abbreviated to $O(\overline{t})[\overline{q}]$;
\item if $c$ is a classical expression of type $\mathbf{C}$ and $A$ a formal operator, then $cA$ (as an abbreviation of $c\cdot A$) is a formal operator, and $\mathit{Sign}(cA)=\mathit{Sign}(A)$;
\item if $A$ is a formal operator with $\mathit{Sign}(A)=\overline{q}\rightarrow\overline{q^\prime}$, then so is $A^\dag$, and $\mathit{Sign}(A^\dag)=\overline{q^\prime}\rightarrow\overline{q}$;
\item if $A_1,A_2$ are formal operators with the same signature, then $A_1+A_2$ is a formal operator, and $\mathit{Sign}(A_1+A_2)=\mathit{Sign}(A_1)$;
\item if $A_1,A_2$ are formal operators with $\mathit{Sign}(A_1)=\overline{q}\rightarrow\overline{q^\prime}$ and $\mathit{Sign}(A_2)=\overline{q^\prime}\rightarrow\overline{q^{\prime\prime}}$, then $A_1A_2$ (as an abbreviation of $A_1\cdot A_2$) is a formal operator, and  $\mathit{Sign}(A_1A_2)=\overline{q}\rightarrow\overline{q^{\prime\prime}}$; 
\item if $A_1,A_2$ are formal operators with $\mathit{Sign}(A_i)=\overline{q_i}\rightarrow\overline{q_i^\prime}$ for $i=1,2$, then $A_1\otimes A_2$ is a formal operator, and $\mathit{Sign}(A_1\otimes A_2)=\overline{q_1},\overline{q_2}\rightarrow\overline{q_1^\prime},\overline{q_2^\prime}$. 
\end{enumerate}
\end{defn}

It is worth noting that signature $\overline{q}\rightarrow\overline{q^\prime}$ in a basic formal operators $O(\overline{t})[\overline{q}\rightarrow\overline{q^\prime}]$ and $X[\overline{q}\rightarrow\overline{q^\prime}]$ is used to indicate the quantum variables in the domain and codomain of operator constant $O$. In the above definition, only $\cdot$ (scalar multiplication), $^\dag$ (adjoint), $+$ (addition), $\cdot$ (multiplication) and $\otimes$ (tensor product) are allowed to be operations symbols (of operators). In the future applications, more operation symbols can be introduced; e.g. $e^A$, $\log A$, $\sqrt{A}$ and tracing out certain quantum variables. For simplicity of the presentation, we often introduce some derived operations of formal operators (as syntactic sugars). For example, if $I$ is a basic operator symbol denoting the identity operator, then for any operator $A$, we write $\neg A:=I-A=I+(-1)A.$

It is easy to see that a formal operator $A$ is not always meaningful. Let $\mathit{Sign}(A)=\overline{q}\rightarrow\overline{q^\prime}$. Then 
$A$ is meaningful only when the quantum variables in $\overline{q}$ are distinct, and so are those in $\overline{q^\prime}$. However, this cannot be determined syntactically, as discussed in Subsection \ref{sec-qvar}. This issue will be resolved by introducing the signing rules in Subsection~\ref{sec-sign}.

The syntax of formal operators defined above allows us to construct more and more complex formal operators from atomic ones like $|s\rangle_q, \langle s|_q, O[\overline{q}\rightarrow\overline{q^\prime}]$ and $X[\overline{q}\rightarrow\overline{q^\prime}]$
using the operations like $^\dag, +, \cdot, \otimes$, etc. Let us give an example to show how it can significantly simplify the specification of quantum data and operations.   

\begin{exam}\label{exam-pred}For any classical bit array $J$ of type $\mathbf{Int}\rightarrow\mathbf{Bool}$, we write segment $J[k:l]$ for the restriction of $J$ to the interval $[k:l]=\{\mbox{integer}\ i|k\leq i\leq l\}$. Let $q$ be a qubit array of the same type $\mathbf{Int}\rightarrow\mathbf{Bool}$. We consider the segment $q[k:l]$ --- a composite system of $l-k+1$ qubits. Then: \begin{enumerate}\item Each basis state of $q[k:l]$can be written as a formal operator $$|J[k:l]\rangle_{q[k:l]}=|J[k]\rangle_{q[k]}\otimes ...\otimes |J[l]\rangle_{q[l]}$$ for some classical bit array $J$. 
\item In particular, the projection operator onto the subspace spanned by all basis states $|J[k:l]\rangle_{q[k:l]}$ with $J[i]\neq 1$ for some $k\leq i\leq l$ can be simply written as $$\neg|J_0[k:l]\rangle_{q[k:l]}\langle J_0[k,l]|$$ where $J_0$ is the bit array such that $J_0[i]=0$ for all $k\leq i\leq l.$\end{enumerate}
\end{exam}

Indeed, constructing complex matrices and operators from simpler ones is a common practice in the quantum information and computation literature. Definition \ref{def-fop} is simply a formalisation of this idea in the language of formal logic.

\subsection{Substitution}\label{sec-fop-sub}

 Substitution is a key technique in symbolic reasoning. It was defined in Definition \ref{def-term-sub} for classical logic $\mathcal{L}$. To accomodate this technique in reasoning about quantum data and operations, let us introduce: 
\begin{defn}[Substitution of Classical Variables] Let $x$ and $t$ be the same as in Definition \ref{def-term-sub}. Then for any formal operator $A$, formal operator $A[t/x]$ is obtained by substitution of $t$ for $x$ in $A$. More precisely, it is inductively defined as follows: \begin{enumerate}\item if $A=c$, then $A[t/x]=c$; \item if $A=|s\rangle_q$, then $A[t/x]=|s[t/x]\rangle_{q[t/x]}$, where $s[t/x]$  is given by Definition \ref{def-term-sub}, and $q[t/x]=q^\prime[s_1[t/x],...,s_n[t/x]]$ if $q=q^\prime[s_1,...,s_n]$ (in particular, $q[t/x]=q$ if $q$ is a simple quantum variable); \item if $A=\langle s|_q$, then $A[t/x]=\langle s[t/x]|_{q[t/x]}$; 
\item if $A=X[\overline{q}\rightarrow\overline{q^\prime}]$, then $A[t/x]=X[\overline{q}[t/x]\rightarrow\overline{q^\prime}[t/x]]$, where $\overline{q}[t/x]=q_1[t/x],...,q_m[t/x]$ if $\overline{q}=q_1,...,q_m$ (in particular, $\epsilon[t/x]=\epsilon$); 
\item if $A=O(\overline{s})[\overline{q}\rightarrow\overline{q^\prime}]$, then $A[t/x]=O(\overline{s}[t/x])[\overline{q}[t/x]\rightarrow\overline{q^\prime}[t/x]]$, where $\overline{s}[t/x]=s_1[t/x],...,s_k[t/x]$ if $\overline{s}=s_1,...,s_k$ (in particular, $\epsilon[t/x]=\epsilon$); 
\item if $A=cA^\prime$, then $A[t/x]=cA^\prime[t/x]$;
\item if $A=A^{\prime\dag}$, then $A[t/x]=(A^\prime[t/x])^\dag$; 
\item if $A=A_1+A_2$, then $A[t/x]=A_1[t/x]+A_2[t/x]$;
\item if $A=A_1A_2$, then $A[t/x]=A_1[t/x]\cdot A_2[t/x]$;
\item if $A=A_1\otimes A_2$, then $A[t/x]=A_1[t/x]\otimes A_2[t/x]$.
\end{enumerate}
\end{defn}

Essentially, the above definition lifts the notion of substitution from the underlying logic $\mathcal{L}$ to the expressions in our symbolic operator logic $\mathbf{SOL}$. Such a lifting will be further extended to logical formulas in $\mathbf{SOL}$. We can also introduce:

\begin{defn}[Substitution of Operator Variables] Let $X$ be an operator variable, and let $A,B$ be formal operators with $\mathit{Sign}(B)=\overline{q}\rightarrow\overline{q^\prime}$. Then formal operator $A[B/X]$ is obtained by substitution of $B$ for $X$ in $A$. Formally, it is inductively as follows: \begin{enumerate}\item if $A=c, |s\rangle_q, \langle s|_q, O(\overline{t})[\overline{r}\rightarrow\overline{r^\prime}]$, or $Y[\overline{r}\rightarrow\overline{r^\prime}]$ with $Y\neq X$, then $A[B/X[\overline{q}\rightarrow\overline{q^\prime}]]=A$;
\item if $A=X[\overline{r}\rightarrow\overline{r^\prime}]$, then $A[B/X]=B[\overline{r},\overline{r^\prime}/\overline{q},\overline{q^\prime}]$ (the substitution of $\overline{r},\overline{r^\prime}$ for $\overline{q},\overline{q^\prime}$ in $B$);
\item $(cA)[B/X]=c A[B/X]$;
\item $(A^\dag)[B/X]=(A[B/X])^\dag$;
\item $(A_1+A_2)[B/X]]=A_1[B/X]+A_2[B/X]$;
\item $(A_1A_2)[B/X]=A_1[B/X]\cdot A_2[B/X]$;
\item $(A_1\otimes A_2)[B/X]=A_1[B/X]\otimes A_2[B/X]$.
\end{enumerate}\end{defn}  

The above two definitions can be straightforwardly generalised to simultaneous substitutions $A[\overline{t}/\overline{x}]$ and $A[\overline{B}/\overline{X}].$ 

\subsection{Assigning Signatures}\label{sec-sign}

 To define the semantics of formal operators, we fix a first-order structure $\mathcal{S}$ as defined in Section \ref{sec-classic}. Let $\sigma\in\mathit{Stat}_\mathcal{S}$ be a classical state in the given structure. As pointed out in Subsection \ref{sec-fop-s}, it is not the case that the semantics of every formal operator $A$ in $\sigma$ is well-defined. To ensure the well-definedness of its semantics, let us first define the signature of $A$ in $\sigma$. 
 For each simple or subscripted quantum variable $q$ of type $T$, $\sigma(q)$ is a ground quantum variable denoting a quantum system with state space $\hs_T$. 
 For a signature $\overline{q}\rightarrow\overline{q^\prime}$, we define $\sigma\left(\overline{q}\rightarrow\overline{q^\prime}\right)=\sigma\left(\overline{q}\right)\rightarrow\sigma\left(\overline{q^\prime}\right)$ as a ground signature. Furthermore, we introduce:
 
 \begin{defn}\begin{enumerate}\item We say that a formal operator $A$ is well-signed in a classical state $\sigma$ if a ground signature judgement of the form:
 \begin{equation}\label{eq-sign-rules}\sigma\models A:\overline{q}\rightarrow\overline{q^\prime}\end{equation} can be inferred from the signing rules in Table \ref{fig-sign-rules}. 
 \begin{table}[t]
\begin{equation*}\begin{split}&(\mbox{Sign-Cons})\quad \sigma\models c:\epsilon\rightarrow\epsilon\qquad (\mbox{Sign-Stat})\quad \sigma\models |s\rangle_q:\sigma(q)\rightarrow\epsilon\quad \sigma\models \langle s|_q:\epsilon\rightarrow\sigma(q)\\ &(\mbox{Sign-OpV})\quad \frac{\sigma\models\mathit{Dist}(\overline{q})\wedge\mathit{Dist}(\overline{q^\prime})}{\sigma\models X[\overline{q}\rightarrow\overline{q^\prime}]:\sigma\left(\overline{q}\rightarrow\overline{q^\prime}\right)} \quad (\mbox{Sign-OpC})\quad \frac{\sigma\models\mathit{Dist}(\overline{q})\wedge\mathit{Dist}(\overline{q^\prime})}{\sigma\models O(\overline{t})[\overline{q}\rightarrow\overline{q^\prime}]:\sigma\left(\overline{q}\rightarrow\overline{q^\prime}\right)}\\ &(\mbox{Sign-Scal})\quad \frac{\sigma\models A: \overline{q}\rightarrow\overline{q^\prime}}{\sigma\models cA: \overline{q}\rightarrow\overline{q^\prime}}\qquad\qquad\qquad\ (\mbox{Sign-Diag})\quad \frac{\sigma\models A: \overline{q}\rightarrow\overline{q^\prime}}{\sigma\models A^\dag: \overline{q^\prime}\rightarrow\overline{q}}\\ &(\mbox{Sign-Add})\quad \frac{\sigma\models A_i: \overline{q}\rightarrow\overline{q^\prime}\ (i=1,2)}{\sigma\models A_1+A_2: \overline{q}\rightarrow\overline{q^\prime}}\qquad\ \ \ (\mbox{Sign-Mult})\quad 
\frac{\sigma\models A_1: \overline{q^\prime}\rightarrow\overline{q^{\prime\prime}}\quad \sigma\models A_2: \overline{q}\rightarrow\overline{q^\prime}}{\sigma\models A_1A_2: \overline{q}\rightarrow\overline{q^{\prime\prime}}}\\ 
&(\mbox{Sign-Tensor})\quad \frac{\sigma\models A_i: \overline{q_i}\rightarrow\overline{q_i^\prime}\ (i=1,2)\quad \overline{q_1}\cap\overline{q_2}=\overline{q_1^\prime}\cap\overline{q_2^\prime}=\emptyset}{\sigma\models A_1\otimes A_2: \overline{q_1};\overline{q_2}\rightarrow\overline{q_1^\prime};\overline{q_2}}
 \end{split}\end{equation*}
\caption{Signing Rules for Formal Operators.}\label{fig-sign-rules}
\end{table}
\item Let $\Sigma$ be a set of classical first-order logic formulas, $A$ a formal operators and $\overline{q},\overline{q^\prime}$ strings of simple or subscripted quantum variables. Then by $$\Sigma\models A:\overline{q}\rightarrow\overline{q^\prime}$$ we mean that for any classical state $\sigma$, whenever $\sigma\models\Sigma$, then $\sigma\models A:\sigma(\overline{q}\rightarrow\overline{q^\prime}).$
\end{enumerate}
 \end{defn}
 
 The following lemma presents a connection between the signatures of a formal operator and its substitutions. 
 \begin{lem}[Signature for Substitution]\label{lem-sub-sign} For any ground signature  $\overline{q}\rightarrow\overline{q^\prime}$, we have: \begin{equation}\label{eq-sign1}\sigma\models A[\overline{t}/\overline{x}]:\overline{q}\rightarrow\overline{q^\prime}\ \mbox{iff}\ \sigma[\overline{x}:=\sigma(\overline{t})]\models A:\overline{q}\rightarrow\overline{q^\prime}.\end{equation}
 \end{lem}
 \begin{proof} We can prove this lemma by induction on the structure of $A$. 
 \end{proof}
 
 \subsection{Semantics}
 
 Now we are ready to present the definition of the semantics of formal operators. Essentially, the semantics of formal operators is obtained by embedding the semantics of classical expressions into it. Given a first-order structure $\mathcal{S}$ of logic $\mathcal{L}$. We first introduce:
 \begin{defn} A quantum structure over the given classical structure $\mathcal{S}$ is a pair $\mathcal{Q}=(\mathcal{H},\mathbb{I})$, where: \begin{enumerate}\item $\hs$ is a family $\left\{\hs_T\right\}$ of Hilbert spaces such that for each basic type $T$, the quantum domain of $T$ is $$\hs_T=\mathit{span}\{|d\rangle\ |\ d\in\mathcal{D}_T\}$$ where $\mathcal{D}_T$ is the classical domain of $T$ in $\mathcal{S}$, and the quantum domain of a higher type $T_1\times ...\times T_n^\prime\rightarrow T$ is defined by equation (\ref{highH});
  \item $\mathbb{I}$ interprets each operator constant $O$ of quantum type $T_1\times...\times T_m\rightarrow T^\prime_1\times...\times T^\prime_n$ and parameter type $T^{\prime\prime}_1\times...\times T^{\prime\prime}_k$ as a family $O^\mathbb{I}=\left\{O^\mathbb{I}(\overline{d})\right\}$ of linear operators from Hilbert space $\bigotimes_{i=1}^m\hs_{T_i}$ to $\bigotimes_{i=1}^n\hs_{T^\prime_i}$, where indices $\overline{d}\in\prod_{i=1}^k\mathcal{D}_{T_i^{\prime\prime}}$. (For simplicity, we abuse the notation $\mathbb{I}$ to denote the interpretations in both the classical and quantum structures.)  
\end{enumerate} \end{defn} 
 
 A valuation $\eta$ of operator variables in the quantum structure $\mathcal{Q}$ is then defined as a mapping from each operator variable $X$ of type $T_1\times ...\times T_m\rightarrow T_1^\prime\times ...\times T_n^\prime$ to a linear operator $\eta(X)$ from Hilbert space  $\bigotimes_{i=1}^m\hs_{T_i}$ to $\bigotimes_{i=1}^n\hs_{T^\prime_i}$. 
 
 \begin{defn} Let $\zeta=(\sigma,\eta)$, where $\sigma$ is a classical state in $\mathcal{S}$ and $\eta$ a valuation of operator variables in $\mathcal{Q}$. Then the semantics $\zeta(A)$ of a formal operator $A$ in $\sigma$ is inductively defined as follows: \begin{enumerate}\item If $A=c$, then $\zeta(A)=c^\mathbb{I}$;
 \item If $A=|s\rangle_q$, then $\zeta(A)=|\sigma(s)\rangle_{\sigma(q)}$ (standard basis state of the quantum system denoted by $\sigma(q)$), where $\sigma(s)$ and $\sigma(q)$ are given by Definitions \ref{def-term-value} and \ref{def-ground};
 \item If $A=\langle s|_q$, then $\zeta(A)=\langle\sigma(s)|_{\sigma(q)}$, where $\sigma(s)$ and $\sigma(q)$ are as in clause (2);
 \item If $A=X[\overline{q}\rightarrow\overline{q^\prime}]$, then $\zeta(A)=\eta(X):\hs_{\sigma(\overline{q})}\rightarrow\hs_{\sigma(\overline{q^\prime})}$; 
  \item If $A=O(\overline{t})[\overline{q}\rightarrow\overline{q^\prime}]$, then $\zeta(A)=O^\mathbb{I}(\sigma(\overline{t})):\hs_{\sigma(\overline{q})}\rightarrow\hs_{\sigma(\overline{q^\prime})}$; 
 \item If $A=cA^\prime$, then $\zeta(A)=c^\mathbb{I}\cdot\sigma(A^\prime)$;
 \item If $A=A^{\prime\dag}$, then $\zeta(A)=\zeta(A^\prime)^\dag$;
 \item If $A=A_1+A_2$, then $\zeta(A)=\zeta(A_1)+\zeta(A_2)$;
 \item If $A=A_1A_2$, then $\zeta(A)=\zeta(A_1)\cdot\zeta(A_2)$;
 \item If $A=A_1\otimes A_2$, then $\zeta(A)=\zeta(A_1)\otimes \zeta(A_2)$.
\end{enumerate} 
\end{defn}

It should be noted that the operations in the right-hand side of the defining equations of $\zeta(A)$ in clauses (6) to (10) are operations of linear operators on the appropriate Hilbert spaces. 

The following lemma shows that the signature of a formal operator $A$ in a classical state $\sigma$ indeed describes the domain and codomain of the semantics $\zeta(A)$.  
 
 \begin{lem} Let $\zeta=(\sigma,\eta)$. If $\sigma\models A:\overline{q}\rightarrow\overline{q^\prime}$, then $\zeta(A)$ is well-defined as a linear operator from $\hs_{\overline{q}}$ to $\hs_{\overline{q^\prime}}$. 
 \end{lem}
 \begin{proof} Routine by induction on the structure of $A$. 
 \end{proof}
 
Thus, the semantics of a formal operator $A$ can be understood as a partial function from pairs $\zeta=(\sigma,\eta)$ of classical states and valuations of operator variables to operator in $\hs_{\overline{q}}\rightarrow \hs_{\overline{q^\prime}}$ with $\sigma\models A:\overline{q}\rightarrow\overline{q^\prime}$:
$$\llbracket A\rrbracket:\zeta=(\sigma,\eta)\mapsto\zeta(A)\ \mbox{whenever}\ \zeta(A)\ \mbox{is well-defined.}$$ 
So semantically $A$ is an operator parameterised by classical variables within it.  
This understanding is consistent with the notion of cq-states and cq-predicate defined in \cite{FengY} and unitary expressions \cite{You25}. A fundamental difference between our approach and those in \cite{FengY, You25} is that we introduced the syntax of parameterised operators. It is noteworthy that the codomain $\hs_{\overline{q}}\rightarrow \hs_{\overline{q^\prime}}$ of the semantic mapping $\llbracket A \rrbracket$ depends on $\sigma$ in inputs $\zeta=(\sigma,\eta)$; in other words, $A$ should be viewed as dependently typed (in the sense of dependent type theory). 

Obviously, if $A$ does not contain any operator variable and $\zeta=(\sigma,\eta)$, then $\zeta(A)$ depends only on $\sigma$ but not $\eta$. In this case, $\zeta(A)$ is often abbreviated to $\sigma(A)$. 

  \begin{exam}Consider formal operator: $$s=\cos\frac{\theta}{2}|x-\frac{1}{2}\rangle_{q[3n-2]}+\sin\frac{\theta}{2}|x+\frac{1}{2}\rangle_{q[3n-2]}.$$ It stands for a quantum state parameterised by classical variables $\theta, x$ and $n$. Let classical state $\sigma$ be given by $\sigma(\theta)=\frac{\pi}{2}, \sigma(x)=\frac{1}{2}, \sigma(n)=3$. Then $\sigma(s)$ is the state $|+\rangle=\frac{1}{\sqrt{2}}(|0\rangle+|1\rangle)$ of qubit $q[7]$.
 \end{exam}

The substitution lemma for classical expressions (Lemma \ref{lem-classical-sub}(1)) can be lifted to formal operators:
\begin{lem}[Substitution]\label{lem-fop-sub} For any classical variable $x$ and expression $t$, for any $\zeta=(\sigma,\eta)$, and for any formal operator $A$, we have: 
\begin{equation*}\begin{split}&\zeta(A[t/x])=(\sigma[x:=\sigma(t)],\eta)(A);\\ 
&\zeta(A[B/X])=(\sigma,\eta[X:=\zeta(B)])(A) \end{split}\end{equation*} where $\eta[X:=\zeta(B)]$ stands for the updated valuation of operator variables defined by $$\eta[X:=\zeta(B)](Y)=\begin{cases}\eta(Y)\ &\mbox{if}\ Y\neq X, \\ \zeta(B)\ &\mbox{if}\ Y=X.\end{cases}$$
\end{lem}
\begin{proof} By induction on the structure of $A$, using Lemmas \ref{lem-classical-sub}(1) and \ref{lem-sub-sign}.\end{proof} 

 \section{Functions and Predicates over Formal Operators}\label{sec-fp}
 
 We recall from Section \ref{sec-classic} that in classical first-order logic $\mathcal{L}$ a function is defined as a mapping from individuals (i.e. elements of the domains) to themselves, and a predicate is a mapping from individuals to truth values. At the layer of symbolic operator logic $\mathbf{SOL}$, operators are treated as individuals and formal operators are considered as expressions. So, in this section, we introduce some basic functions and predicates over formal operators upon which logical formulas of $\mathbf{SOL}$ will be constructed. 
 
 \subsection{Functions over Formal Operators}

 Formal operators defined in the last section are essentially syntactic representations of matrices and operators (parameterised by classical variables). Then various functions over matrices and operators used in quantum computation and information  \cite{NC00} can be naturally generalised to formal operators. In this section, we only present several examples of them to illustrate how such generalisations can be done. More functions over formal operators can be introduced in the future applications.  
 
  \begin{exam}\label{exam-norm} We can define the function of Frobenius norm $\|A\|$ for formal operators $A$. For a given first-order structure $\mathcal{S}$ as the interpretation of classical logic $\mathcal{L}$ and for a state $\sigma$ in $\mathcal{S}$, and for a quantum structure $\mathcal{Q}$ over $\mathcal{S}$ and a valuation $\eta$ of operator variables in $\mathcal{Q}$, let $\zeta=(\sigma,\eta)$. Then the semantics $\zeta(\|A\|)$ of $\|A\|$ in $\zeta$ is defined as follows:
  \begin{equation*}\zeta(\|A\|)=\|\zeta(A)\|=\sqrt{\sum_{i,j}|a_{ij}|^2}\ \mbox{if}\ \zeta(A)\ \mbox{is matrix}\ (a_{ij}).
  \end{equation*}
   \end{exam}
  
 \begin{exam} The trace $\tr(A)$ of a formal operator $A$ can be defined in the same way as in Example \ref{exam-norm}. For a given pair $\zeta=(\sigma,\eta)$ of classical state $\sigma$ and valuation $\eta$ of operator variables,  \begin{equation*}\zeta(\tr(A))=\begin{cases}\tr(\zeta(A))\ &\mbox{if}\ \sigma\models A:\overline{q}\rightarrow\overline{q}\ \mbox{for some}\ \overline{q},\\ \mbox{undefined}\ &\mbox{otherwise}.\end{cases}\end{equation*}
\end{exam}

Various properties of operators can be generalised to formal operators, for instance linearity of trace: $\tr(\lambda A+\mu B)=\lambda\tr(A)+\mu\tr(B).$

 \subsection{Predicates over Formal Operators}
 
Symbolic operator logic $\mathbf{SOL}$ is a first-order logic with operators as individuals (i.e. objects), and formal operators defined in Section \ref{sec-operators} are expressions (i.e. terms) in $\mathbf{SOL}$. In the previous subsection, we show how to define functions over formal operators. In this subsection, we further illustrate how predicates over formal operators can be defined as atomic propositions in $\mathbf{SOL}$ by presenting some simple examples. 
 
\begin{exam}[Quantum states, Unitary Transformations and Observables] We introduce unary predicates $\mathcal{S}_p$, $\mathcal{S}_m$, $\mathcal{U}$ and $\mathcal{O}$ over formal operators for pure quantum states, mixed quantum states, unitary transformations and observables, respectively. Their semantics are defined as follows. Given a classical structure $\mathcal{S}$ and a quantum structure $\mathcal{Q}$ over it. Let $\sigma$ be a classical state in $\mathcal{S}$ and $\eta$ a valuation of operator variables in $\mathcal{Q}$, and $\zeta=(\sigma,\eta)$, let $A$ be a formal operator, and let $\overline{q}$ be a sequence of simple or subscripted quantum variables.  
 \begin{enumerate}\item The judgement $\zeta\models\mathcal{S}_p(A):\overline{q}$ means that in the context $\zeta$, $A$ denotes a pure state of quantum variables $\overline{q}$; that is, $\sigma\models\mathit{Dist}(\overline{q})$, $\sigma\models A:\sigma(\overline{q})\rightarrow\epsilon$ and $\zeta(A)$ is a unit vector in Hilbert space $\hs_{\sigma(\overline{q})}$.  
 
 \item The judgement $\zeta\models\mathcal{S}_m(A):\overline{q}$ means that in the context $\zeta$, $A$ denotes a mixed state of quantum variables $\overline{q}$; that is, $\sigma\models\mathit{Dist}(\overline{q})$, $\sigma\models A:\sigma(\overline{q})\rightarrow\sigma(\overline{q})$ and $\zeta(A)$ is a density operator (i.e. a positive operator with trace $1$) on Hilbert space $\hs_{\sigma(\overline{q})}$.  
 
 \item The judgement $\zeta\models\mathcal{U}(A):\overline{q}$ means that in the context $\zeta$, $A$ denotes a unitary transformation on quantum variables $\overline{q}$; that is, $\sigma\models\mathit{Dist}(\overline{q})$, $\sigma\models A:\sigma(\overline{q})\rightarrow\sigma(\overline{q})$ and $\zeta(A)$ is a unitary operator on Hilbert space $\hs_{\sigma(\overline{q})}$.  
 
 \item The judgement $\zeta\models\mathcal{O}(A):\overline{q}$ means that in the context $\zeta$, $A$ denotes an observable on quantum variables $\overline{q}$; that is, $\sigma\models\mathit{Dist}(\overline{q})$, $\sigma\models A:\sigma(\overline{q})\rightarrow\sigma(\overline{q})$ and $\zeta(A)$ is a Hermitian operator on Hilbert space $\hs_{\sigma(\overline{q})}$.  
\end{enumerate}
 \end{exam}
 
The predicates over formal operators defined in the above example can be viewed from a different perspective. As discussed in the introduction, the framework is deliberately designed to treat all formal operators as a single class of syntactic entities within our logical system for specifying quantum data and operations. However, the preceding example refines this framework by explicitly distinguishing among pure and mixed quantum states, unitary transformations, and observables -- an approach that is often required in practical applications.

\begin{exam}We can introduced binary predicates $=$ (equality) and $\sqsubseteq$ (L\"{o}wner order) over formal operators. For any two formal operators $A$ and $B$, and for any pair $\zeta=(\sigma,\eta)$ of classical state $\sigma$ and valuation $\eta$ of operator variables:\begin{enumerate}\item The judgement $\zeta\models A=B:\overline{q}\rightarrow\overline{q^\prime}$ is true if $\sigma\models\mathit{Dist}(\overline{q})\wedge\mathit{Dist}(\overline{q^\prime})$, $\sigma\models A:\sigma(\overline{q})\rightarrow\sigma(\overline{q^\prime})$, $\sigma\models B:\sigma(\overline{q})\rightarrow\sigma(\overline{q^\prime})$ and $\zeta(A)=\zeta(B)$. 
\item The judgement $\zeta\models A\sqsubseteq B:\overline{q}$ is true if $\sigma\models\mathit{Dist}(\overline{q})$, $\sigma\models A:\sigma(\overline{q})\rightarrow\sigma(\overline{q})$, $\sigma\models B:\sigma(\overline{q})\rightarrow\sigma(\overline{q})$ and $\zeta(B)-\zeta(A)$ is a positive operator.
\end{enumerate}
\end{exam}

The predicate symbols over operators introduced in the above two examples are some of the mostly basic ones. More predicate symbols can be introduced for different applications.
   
\section{Symbolic Operator Logic}\label{sec-sol}

Formal operators are defined as expressions (i.e. terms of individuals) in symbolic operator logic $\mathbf{SOL}$ in Section \ref{sec-operators}. Upon it, we are now ready to define the syntax and semantics of the logic system $\mathbf{SOL}$ as well as the notion of entailment in it. 

\subsection{Syntax}
With the basic functions and predicates presented in Section \ref{sec-fp}, we can define atomic propositions in $\mathbf{SOL}$. From these, more complex logical formulas in $\mathbf{SOL}$ can be constructed using logical connectives and quantifiers.

\begin{defn}[Syntax] Logical formulas over operators (i.e. well-formed formulas in $\mathbf{SOL}$) are defined by the following syntax:
\begin{align}\label{logic1}\mathcal{A}: =\ &\|A\|\propto\lambda\ |\ \tr(A)\propto\lambda\ |\ ...\\ 
 \label{logic2}&|\ \mathcal{S}_p(A):\overline{q}\ |\ \mathcal{S}_m(A):\overline{q}\ |\ \mathcal{U}(A):\overline{q}\ |\ \mathcal{O}(A):\overline{q}\ |\ \cdots \\ \label{logic3}& |\ A_1=A_2 \ | \ A_1\sqsubseteq A_2\ |\ \cdots \\ \label{logic4}& |\ \neg\mathcal{A}\ |\ \mathcal{A}_1\wedge\mathcal{A}_2\\ \label{logic5} & |\ (\forall x)\mathcal{A}\ |\ (\forall X)\mathcal{A} \end{align} where $\lambda$ is a constant, $\propto\in\{=,<,>\}$, $x$ is a classical variable and $X$ an operator variable. 
\end{defn}

Some basic atomic propositions are given in (\ref{logic1}) - (\ref{logic3}). Additional atomic propositions may be introduced in future applications. The propositional connectives $\neg$ (negation) and $\wedge$ (conjunction) are introduced in (\ref{logic4}). The universal quantifiers $(\forall x)$ over classical variables and $(\forall X)$ over operator variables are introduced in (\ref{logic5}).   
Other logical connectives -- $\vee$ (disjunction), $\rightarrow$ (implication), $\leftrightarrow$ (bi-implication)  -- as well as the existential quantifier $\exists$ can be defined in terms of $\neg, \vee$, and $\forall$ in the familiar way. We write $\mathit{Wff}_\mathbf{SOL}$ for the set of all Logical formulas over operators. 

Let us present several simple examples to show how the syntax of symbolic operator logic $\mathbf{SOL}$ defined above can be conveniently used to specify various properties of quantum data and operations.  

 \begin{exam}[Effects = Quantum Predicates] Quantum predicates were introduced in \cite{DP06} as the preconditions and postconditions in quantum Hoare triples (i.e. correctness formulas for quantum programs). They are defined as Hermitian operators (i.e. observables) between zero operator $0$ and identity operator $I$. Indeed, they are called effects in the quantum foundations literature. Formally, they can be characterised in the language of $\mathbf{SOL}$ as follows: $$\mathcal{P}(A):\overline{q}::=\mathcal{O}(A):\overline{q}\wedge 0\sqsubseteq A\sqsubseteq I.$$ Here, $\mathcal{P}(A)$ stands for the statement that $A$ is a quantum predicate over quantum variables $\overline{q}$. 
  \end{exam}
  
  \begin{exam}[No-Cloning Theorem]\label{no-clone} No-cloning theorem asserts that there is no copying machine that can create an independent and identical copy of an arbitrary unknown quantum state (see for example \cite{NC00}, Section 1.3.5). Let $q_1,q_2$ be two quantum variables with the same type. Then no-cloning theorem can be formally stated in $\mathbf{SOL}$ as follows:
  \begin{equation}\neg(\exists U)\left[\mathcal{U}(U):q_1,q_2\wedge (\forall |\psi\rangle)\left(\mathcal{S}_p(|\psi\rangle):q_1\rightarrow U(|\psi\rangle\otimes|0\rangle)=|\psi\rangle\otimes|\psi\rangle\right)\right].
  \end{equation}
\end{exam}

The substitutions $A[\overline{t}/\overline{x}]$ and $A[\overline{B}/\overline{X}]$ for formal operators $A$ defined in Subsection \ref{sec-fop-sub} can be easily extended to substitutions $\mathcal{A}[\overline{t}/\overline{x}]$ and $\mathcal{A}[\overline{B}/\overline{X}]$ for logical formulas $\mathcal{A}\in\mathit{Wff}_\mathbf{SOL}$ by employing the familiar technique of $\alpha$-conversion (see clause (2) of Definition \ref{def-c-sub}.  

\subsection{Semantics}

To define the semantics of $\mathbf{SOL}$, we assume that a classical first-order structure $\mathcal{S}$ and a quantum structure $\mathcal{Q}$ over it are given.
 
\begin{defn}[Semantics] Let $\zeta=(\sigma,\eta)$, where $\sigma$ is a classical state in $\mathcal{S}$ and $\eta$ a valuation of operator variables in $\mathcal{Q}$, and let $A\in\mathit{Wff}_\mathbf{SOL}$ be a logical formula over operators. Then satisfaction $\zeta\models \mathcal{A}$ is defined by induction on the structure of $ \mathcal{A}$ as follows: 
\begin{enumerate}\item If $ \mathcal{A}$ is an atomic formula, then $\zeta\models  \mathcal{A}$ was defined in Section \ref{sec-fp};
\item $\zeta\models \neg\mathcal{A}$ iff it is not true that $\zeta\models\mathcal{A}$;
\item $\zeta\models \mathcal{A}_1\wedge \mathcal{A}_2$ iff $\zeta\models \mathcal{A}_1$ and $\zeta\models \mathcal{A}_2$;
\item $\zeta\models(\forall x) \mathcal{A}$ iff $(\sigma[x:=d],\eta)\models \mathcal{A}$ for any value $d$ of the same type as $x$; and
\item $\zeta\models (\forall X)\mathcal{A}$ iff $(\sigma,\eta[X:=A])\models  \mathcal{A}$ for any operator $A$ of the same type of $X$.  
\end{enumerate}
\end{defn}

Let $\Gamma\subseteq\mathit{Wff}_\mathbf{SOL}$ be a logical formulas over operators. Then by $\zeta\models\Gamma$, we mean that $\zeta\models\mathcal{A}$ for every $\mathcal{A}\in\Gamma$. 

The substitution lemma for classical first-order logic $\mathcal{L}$ (Lemma \ref{lem-classical-sub}) can be lifted to symbolic operator logic $\mathbf{SOL}$. 

\begin{lem}[Substitution]\label{sol-sub} For any logical formula $\mathcal{A}\in\mathit{Wff}_\mathbf{SOL}$ and for any $\zeta=(\sigma,\eta)$ of classical state $\sigma$ and valuation $\eta$ of operator variables, we have: 
\begin{equation}\zeta\models\mathcal{A}[\overline{t}/\overline{x}][\overline{A}/\overline{X}]\ \mbox{iff}\ (\sigma[\overline{x}:=\sigma(\overline{t}],\eta[\overline{X}:=\zeta(\overline{A})])\models\mathcal{A}.
\end{equation}
\end{lem}
\begin{proof} By induction on the structure of $\mathcal{A}$, using Lemma \ref{lem-fop-sub}. 
\end{proof}

\subsection{Entailment}

Based on the semantics of $\mathbf{SOL}$, we can introduce the notion of entailment from a set $\Gamma$ of logical formulas in $\mathbf{SOL}$ to a logical formula $\mathcal{A}$ in $\mathbf{SOL}$ with a set $\Sigma$ of classical first-order logical formulas as an assumption.  

\begin{defn}[Entailment] Let $\Sigma\subseteq\mathit{Wff}_\mathcal{L}$ be a set of classical first-order logical formulas, $\Gamma\subseteq \mathit{Wff}_\mathbf{SOL}$ a set of logical formulas over operators, and $\mathcal{A}\in \mathit{Wff}_\mathbf{SOL}$. Then we say that $\Sigma$ and $\Gamma$ entail $\mathcal{A}$, written \begin{equation}\label{entail}\Sigma,\Gamma\models\mathcal{A}\end{equation} if for any classical structure $\mathcal{S}$ and any quantum structure $\mathcal{Q}$ over it, for all classical state $\sigma$ in $\mathcal{S}$, and for all valuations $\eta$ of operator variables in $\mathcal{Q}$, whenever $\sigma\models\Sigma$ and $(\sigma,\eta)\models \Gamma$, then we have $\sigma\models\mathcal{A}$. 
\end{defn}

If particular, if $\Sigma=\emptyset$, then (\ref{entail}) can be abbreviated to $\Gamma\models\mathcal{A}$; if $\Gamma=\emptyset$, then (\ref{entail}) can be abbreviated to $\Sigma\models \mathcal{A}$, and $\Sigma=\Gamma=\emptyset$, then we say that $\mathcal{A}$ is logically valid and simply write $\models\mathcal{A}$. 

\begin{exam}[Block sphere visualisation of a qubit \cite{NC00}, page 15]
\begin{align*}|\alpha|^2+|\beta|^2=1\models (\exists\theta,\varphi,\gamma)\left(\alpha|0\rangle+\beta|1\rangle=e^{i\gamma}\left(\cos\frac{\theta}{2}|0\rangle+e^{i\varphi}\sin\frac{\theta}{2}|1\rangle\right)\right)
\end{align*}
\end{exam}

Various valid logical formulas in classical first-order logic can be generalised straightforwardly to symbolic operator logic $\mathbf{SOL}$; for example, 
\begin{enumerate}\item $\models (\forall x)\mathcal{A}\rightarrow \mathcal{A}[t/x]$ for any classical variable $x$ and classical expression $t$; 
\item $\models (\forall X)\mathcal{A}\rightarrow \mathcal{A}[B/X]$ for any operator variable $X$ and formal operator formula $B$;
\item $\models (\forall x)(\mathcal{A}\rightarrow\mathcal{B})\rightarrow (\mathcal{A}\rightarrow (\forall x)\mathcal{B})$ if classical variable $x$ is free in $\mathcal{A}$;  
\item $\models (\forall X)(\mathcal{A}\rightarrow\mathcal{B})\rightarrow (\mathcal{A}\rightarrow (\forall X)\mathcal{B})$ if operator variable $X$ is free in $\mathcal{A}$. 
\end{enumerate}
The inference rules in classical first-order logic, e.g. modus ponens and generalisation, can be extended to $\mathbf{SOL}$. The validity of these logical formulas and the soundness of these inference rules can be easily checked by definition. Furthermore, we have:  
\begin{thm}[Deduction theorem] For any $\Sigma\subseteq\mathit{Wff}_\mathcal{L}$ and $\Gamma\subseteq\mathit{Wff}_\mathbf{SOL}$, and for any $\mathcal{A},\mathcal{B}\in\mathit{Wff}_\mathbf{SOL}$, it holds that  
$$\Sigma,\Gamma\cup\{\mathcal{A}\}\models\mathcal{B}\ \mbox{iff}\ \Sigma,\Gamma\models\mathcal{A}\rightarrow\mathcal{B}.$$
\end{thm}
\begin{proof} Easy by the definition of entailment.\end{proof}

\begin{thm}[Substitution]\label{thm-sub} If $\Sigma,\Gamma\models\mathcal{A},$ then for any classical variables $\overline{x}$, operator variables $\overline{X}$, classical expressions $\overline{t}$ and formal operator formulas $\overline{B}$, we have: $$\Sigma[\overline{t}/\overline{x}],\Gamma[\overline{t}/\overline{x}][\overline{B}/\overline{X}]\models\mathcal{A}[\overline{t}/\overline{x}][\overline{B}/\overline{X}].$$
\end{thm}
\begin{proof} Immediate from Lemma \ref{sol-sub}.\end{proof}

\section{Symbolic Reasoning about Quantum Data and Operations}\label{sec-exam1}

In this section, we present a series of examples to show how the logical mechanism developed in the previous section can be employed for symbolic reasoning about quantum data and operations. In some of these examples, In some examples, an entailment $\Sigma,\Gamma\models\mathcal{A}$ will be represented by the following diagram:
$$\frac{\Sigma\quad \Gamma}{\mathcal{A}}$$

\subsection{Basic Properties of Quantum Data and Operations}

Let us first show how some basic properties of quantum data and operations can be specified and derived in the symbolic operator logic $\mathbf{SOL}$.

\begin{exam}The formulations of some basic axioms for quantum states, unitary transformations and observables in the logic $\mathbf{SOL}$ are presented in Table \ref{fig-inf-rules}.
 \begin{table}[t]
\begin{equation*}\begin{split}
&\mbox{(Stat1)}\ \models\mathcal{U}(A_1)\wedge\mathcal{S}_p(A_2)\rightarrow \mathcal{S}_p(A_1A_2)\quad\mbox{(Stat2)}\ \models \mathcal{U}(A_1)\wedge\mathcal{S}_m(A_2)\rightarrow\mathcal{S}_m(A_1A_2A_1^\dag)\\
&\mbox{(Stat3)}\ \models 0\sqsubseteq A\wedge \tr(A)=1\rightarrow\mathcal{S}_m(A)\\ 
&\mbox{(Uni1)}\ \models\mathcal{U}(A)\rightarrow \mathcal{U}(A^\dag)\qquad\qquad\qquad\ \ \ \ \mbox{(Uni2)}\ \models\mathcal{U}(A_1)\wedge\mathcal{U}(A_2)\rightarrow\mathcal{U}(A_1A_2)\\  &\mbox{(Obs1)}\ \models\mathcal{O}(A)\rightarrow A=A^\dag\qquad\qquad\qquad\ \ \ \ \mbox{(Obs2)}\ \models\mathcal{U}(A_1)\wedge\mathcal{O}(A_2)\rightarrow\mathcal{O}(A_1A_2A_1^\dag)\\
&\mbox{(Obs3)}\ \models\mathcal{O}(A_1)\wedge\mathcal{O}(A_2)\wedge A_1A_2=A_2A_1\rightarrow\mathcal{O}(A_1A_2)
\end{split}\end{equation*}
\caption{Axioms for Quantum States, Unitary Transformations and Observables.}\label{fig-inf-rules}
\end{table}
\end{exam}

\begin{exam}The following are the formalisations of the definitions of unitary transformations and observables in $\mathbf{SOL}$: \begin{align}\label{def-uni}\bigwedge_{i=1}^n\left(\sum_{l=1}^n|a_{il}|^2=1\right)\wedge\bigwedge_{1\leq i<j\leq n}\left(\sum_{l=1}^na_{il}a_{jl}^\ast=0\right)&\models\mathcal{U}\left(\sum_{i,j=1}^na_{ij}|i\rangle\langle j|\right)\\
\label{def-obs}\bigwedge_{i,j=1}^n\left(a_{ij}=a_{ji}^\ast\right) &\models \mathcal{O}\left(\sum_{i,j=1}^na_{ij}|i\rangle\langle j|\right).
\end{align} It is worth noting that the left-hand sides of (\ref{def-uni}) and (\ref{def-obs}) are classical first-order logical formulas, whereas the right-hand sides are logical formulas in $\mathbf{SOL}$. Since it holds that $$\mbox{The theory of trigonometric functions}\models |\cos\frac{\theta}{2}|^2+|\sin\frac{\theta}{2}|^2=1$$ we can derive the unitarity of rotations about axes $x, y$ and $z$: $$\models\mathcal{U}(R_x(\theta))\wedge \mathcal{U}(R_y(\theta))\wedge \mathcal{U}(R_z(\theta))$$ from (\ref{def-uni}). Similarly, we can also derive the unitarity of quantum Fourier transform in $\mathbf{SOL}$:
$\models\mathcal{U}(\mbox{QFT}).$
\end{exam}

\subsection{(Conditional) Equational Reasoning}

Since $\mathbf{SOL}$ is a logic with equality, equational reasoning about operators can be realised in it. In particular, whenever the involved  operators carry some classical variables as their parameters, such equational reasoning can be conducted conditional on a classical first-order theory about these variables. The following are some simple examples:  

\begin{exam}\begin{enumerate}\item Linearity: $\models\alpha\left(\sum_i A_i\right)=\sum_i \alpha A_i.$

\item Equality of quantum states: $\models\bigwedge_{i<j} \left(s_i\neq s_j\right)\wedge \left(\sum_i\alpha|s_i\rangle=\sum_i\beta_i|s_i\rangle\right)\rightarrow\bigwedge_i\left(\alpha_i=\beta_i\right).$

\item Coefficient Addition: \begin{align*}\frac{s_1=s_2\wedge q_1=q_2}{\alpha_1|s_1\rangle_{q_1}+\alpha_2|s_2\rangle_{q_2}=(\alpha_1+\alpha_2)|s_1\rangle_{q_1}}
\end{align*}

\item Self Outer-Product: \begin{align*}\frac{s_2=s_3\wedge q_2=q_3}{\left(|s_1\rangle_{q_1}\left(|s_2\rangle_{q_2}^\dag\right)\right)|s_3\rangle_{q_3}=|s_1\rangle_{q_1}}
\end{align*}

\item Identity: \begin{align*}\frac{\dim\overline{q}=m\quad (\forall 1\leq i<j\leq m)(\overline{s}_i\neq \overline{s}_j)}{\sum_{i=1}^m|\overline{s}_i\rangle_{\overline{q}}\langle \overline{s}_i|_{\overline{q}}=I[\overline{q}]}
\end{align*}

\item Matrix Representation: 
$$\frac{A:\overline{q}\rightarrow\overline{q^\prime}\quad \dim\overline{q}=m\quad\dim\overline{q^\prime}=n\quad \overline{s}_i\neq\overline{s}_j\ (1\leq i< j\leq m)\quad \overline{s^\prime}_k\neq\overline{s^\prime}_l\ (1\leq k< l\leq n)}{A=\sum_{i=1}^m\sum_{k=1}^n\langle \overline{s}_i|A|\overline{s^\prime}_k\rangle |\overline{s}_i\rangle_{\overline{q}}\langle\overline{s^\prime}_k|_{\overline{q^\prime}}}$$
\end{enumerate}
\end{exam}

The common practice of labelled Dirac notations in quantum physics and quantum information theory are formalised in \cite{Zhou23} for convenient handling of multiple quantum variables in the verification of quantum programs.
Furthermore, a series of rewriting rules for labelled Dirac notations are collected and implemented in \cite{Xu-25, Xu-25a}. It is obvious that labelled Dirac notations can be easily recasted using the data structure of quantum arrays introduced in Section \ref{sec-arrays}. Thus, the rewriting rules in \cite{Xu-25, Xu-25a} can be generalised into the framework of $\mathbf{SOL}$. What's more provided by $\mathbf{SOL}$ than labelled Dirac notation \cite{Zhou23,Xu-25,Xu-25a} is that one can do address arithmetic for quantum registers or memories and reasoning about it conditional on certain assumptions about the classical variables involved in it expressed in the classical first-order logical language; see (\ref{eq-address}) for a simple example. 

\subsection{(Conditional) Inequational Reasoning}

Since the logic $\mathbf{SOL}$ is also equipped with binary predicate $\sqsubseteq$, inequational reasoning about symbolic operators can also be conveniently carried out in it, possibly conditional on certain assumptions about classical parameters in the involved operators. Here is a simple example:  

\begin{exam} Probabilistic Combination: \begin{align*}\frac{(\forall i)\left(0\leq p_i\leq p_i^\prime\wedge 0\leq A_i\sqsubseteq A_i^\prime\right)}{\sum_ip_iA_i\sqsubseteq\sum_ip_i^\prime A_i^\prime}
\end{align*}
\end{exam}

Some basic techniques for equational and inequational reasoning about symbolic operators are given in the following two propositions: 

\begin{prop} \begin{enumerate}
\item Anti-symmetry and transitivity: \begin{enumerate}\item $\Sigma, \Gamma\models A\sqsubseteq B$ and $\Sigma, \Gamma\models B\sqsubseteq A$ if and only if $\Sigma, \Gamma\models A\equiv B$;
\item If $\Sigma, \Gamma\models A\sqsubseteq B$ and $\Sigma, \Gamma\models B\sqsubseteq C$, then $\Sigma, \Gamma\models A\sqsubseteq C.$
\end{enumerate}
\item Inequality preserved by operations: If $\Sigma, \Gamma\models A\sqsubseteq B$, then: \begin{enumerate}\item  $\Sigma, \Gamma\models\neg cA\sqsubseteq \neg cB$ if $c\geq 0$ and $\Sigma, \Gamma\models cB\sqsubseteq cA$ 
if $c\leq 0$;
\item $\Sigma, \Gamma\models A^\dag\sqsubseteq B^\dag$; 
\item $\Sigma, \Gamma\models A+ C\sqsubseteq B + C$ and $\Sigma, \Gamma\models C+A\sqsubseteq C+B$; 
 \item $\Sigma, \Gamma\models AC\sqsubseteq BC$ and $\Sigma, \Gamma\models CA\sqsubseteq BC$; 
\item $\Sigma, \Gamma\models A\otimes C\sqsubseteq B\otimes C$ and $\Sigma, \Gamma\models C\otimes A\sqsubseteq C\otimes B$. 
\end{enumerate}\end{enumerate}\end{prop}
\begin{proof} Routine by definition. 
\end{proof}

As a corollary of Theorem \ref{thm-sub}, we have:
\begin{prop}[Substitution]\begin{enumerate} \item If $\Sigma, \Gamma\models A\sqsubseteq B$, then $\Sigma[\overline{t}/\overline{x}], \Gamma[\overline{t}/\overline{x}]\models A[\overline{t}/\overline{x}]\sqsubseteq B[\overline{t}/\overline{x}].$
\item If $\Sigma,\Gamma\models A= B$, then $\Sigma[\overline{t}/\overline{x}], \Gamma[\overline{t}/\overline{x}]\models A[\overline{t}/\overline{x}]= B[\overline{t}/\overline{x}].$\end{enumerate}
The same holds for the substitution of operator variables. 
\end{prop}

\section{Recursive Definition of Quantum Data and Quantum Predicates}\label{sec-exam2}

One of the most important applications of symbolic operator logic $\mathbf{SOL}$ with classical first-order logic $\mathcal{L}$ embedded in it is that recursive definitions of quantum data and operators can be conveniently and elegantly written in $\mathbf{SOL}$ where variables in $\mathcal{L}$ are employed as parameters. In this section, we present several simple examples of them, where a recursively defined classical data structure is given, and then a quantum data structure can be recursively defined by a lifting from $\mathcal{L}$ to $\mathbf{SOL}$. 

\begin{exam}[Equal Superposition and GHZ States] Consider qubit array section $q[m:n]$. Then the equal superposition and the GHZ (Greenberger-Horne-Zeilinger) state over it can be seen as quantum states parameterized by classical variables $m$ and $n$, and they are recursively defined as follows: 
  \begin{align*}&|S(m,n)\rangle=\begin{cases}\frac{1}{\sqrt{2}}\left(|0\rangle_{q[m]}+|1\rangle_{q[m]}\right)\ &\mbox{if}\ n=m,\\ |S(m,n-1)\rangle\otimes\frac{1}{\sqrt{2}}\left(|0\rangle_{q[n]}+|1\rangle_{q[n]}\right)\ &\mbox{if}\ n>m;\end{cases}\\  
  &|\mathit{GHZ}(m,n)\rangle=\begin{cases}\frac{1}{\sqrt{2}}\left(|0\rangle_{q[m]}+|1\rangle_{q[m]}\right)\ &\mbox{if}\ n=m,\\ \mathit{CNOT}[q[n-1],q[n]]\left(|\mathit{GHZ}(m,n-1)\rangle\otimes|0\rangle_{q[n]}\right)\ &\mbox{if}\ n>m.\end{cases}\end{align*}
  
It is easy to verify the entailment $m=n\models |S(m,n)\rangle=|\mathit{GHZ}(m,n)\rangle$.
\end{exam}

\begin{exam}[Data Structure in Quantum Fourier Transform] Let $j$ be a classical bit array of type $\mathbf{Int}\rightarrow\mathbf{Bool}$. For any integers $k,l$ with $1\leq k\leq l$, we use $0.j[k:l]$ to denote the binary representation $$0.j_k ...j_l=\sum_{r=k}^l j[r]\cdot 2^{k-r-1}.$$ Then quantum arrays $|(j,k:l)\rangle$ and $|\mathit{QFT}(j,k:l)\rangle$ are defined as states of qubits $q[k:l]$ recursively:  
\begin{align*}&\begin{cases}|(j,k:l)\rangle=|j[l]\rangle\ \mbox{if}\ k=l,\\ |(j,k:r+1)\rangle=|(j,k:r)\rangle\otimes|j[r+1]\rangle\ \mbox{for}\ k\leq r\leq l-1;\end{cases}\\ 
&\begin{cases}|\mathit{QFT}(j,k:l)\rangle=\frac{1}{\sqrt{2}}\left(|0\rangle+e^{2\pi i 0.j[l]}|1\rangle\right)\ \mbox{if}\ k=l,\\ 
|\mathit{QFT}(j,k:r+1)\rangle=|\mathit{QFT}(j,k:r)\rangle\otimes\frac{1}{\sqrt{2}}\left(|0\rangle+e^{2\pi i 0.j[l-r:l]}|1\rangle\right)\ \mbox{for}\ k\leq r\leq l-1. 
\end{cases}
\end{align*} It is easy to see from equation (5.4) in \cite{NC00} that the quantum Fourier transform on $l$ qubits can be rewritten as: 
$$|(j,1:l)\rangle\mapsto |\mathit{QFT}(j,1: l)\rangle.$$
\end{exam}

Many more quantum data structures can be recursively defined in a similar way; for example, regular language quantum states \cite{regular} and graph states \cite{graph}. 

Furthermore, the reader can find some interesting examples of recursively defined quantum gates and quantum algorithms in \cite{Ying24a, Ying24b}. In particular, symbolic operator logic $\mathbf{SOL}$ can serve as an assertion logic in the formal verification of quantum programs. 

\section{Potential Applications}\label{sec-exam3}

In this section, we discuss several other potential applications of the logical framework developed in this paper for specification and reasoning about quantum data and operations. 

\subsection{A Proof Assistant for Quantum Information Theorists} 

The Lean-QuantumInfo library was recently developed and successfully applied to verify a series of theorems in quantum information theory; in particular the generalised quantum Stein's lemma \cite{Stein}. It seems that symbolic operator logic $\mathbf{SOL}$ defined in this paper can be employed as a formal specification language for this kind of libraries. To see this, let us consider a simple example, namely the formalisation of the no-cloning theorem presented in \cite{Stein}. It is not immediately clear from that formalisation that it indeed represents the no-cloning theorem, and its form is somewhat inconvenient for practical applications. By contrast, using the language of $\mathbf{SOL}$, the no-cloning theorem can be formalised as in Example \ref{no-clone}. Obviously, this new formalisation captures the intuition behind the no-cloning theorem more effectively, and can be directly called in applications. Moreover, by integrating the Lean-QuantumInfo library with the reasoning capabilities of $\mathbf{SOL}$, one could develop a powerful proof assistant tailored for quantum information theorists.

\subsection{Specification and Verification of Quantum Programs}

Various quantum program logics have been developed for reasoning about the correctness of quantum programs \cite{Ying11, Unruh19a, FengY, Deng1, QSL2, QSL1, QSL3} and the security of quantum cryptographic protocols \cite{Unruh19, Barthe}. However, as discussed in the Introduction, assertion logics that can specify properties of quantum data and integrate smoothly with these program logics to enable efficient and scalable verification are still lacking. Several attempts have been made to define such assertion languages, including the author’s previous work \cite{Ying22, Ying24a}. Indeed, this paper can be viewed as a continuation and extension of the assertion language proposed in \cite{Ying24a}. We expect that $\mathbf{SOL}$, together with quantum program logics, can be implemented in proof assistants such as Lean and Coq for practical and scalable verification of quantum programs. 

\subsection{Symbolic Execution of Quantum Programs}

Symbolic execution has recently been extended to the analysis of quantum programs \cite{QSE1,QSE3,Fang24}. However, its full potential in quantum computing has not yet been realised, because only very restricted classes of quantum states and operations -- such as stabiliser states and Clifford operators in \cite{Fang24} -- can be properly symbolised in these works. This limitation stems largely from the absence of a general logical framework for specifying and reasoning about arbitrary quantum operators and states. The symbolic operator logic $\mathbf{SOL}$ proposed in this paper appears well suited to be incorporated into symbolic execution, potentially enabling a more expressive and powerful approach to quantum program analysis and verification.

\subsection{Symbolic Model Checking of Quantum Systems} 

Model-checking techniques have also been applied to the analysis and verification of quantum systems, including quantum circuits, quantum programs, and quantum communication protocols. In fact, several model checkers for quantum systems have been implemented \cite{MC1,MC2,MC3}. However, these tools can handle only toy examples involving a small number of qubits.

Symbolic model checking uses logical formulas and decision-diagram–based (DD-based) data structures to represent sets of states and transitions, enabling the reachable state space to be explored much more efficiently by manipulating large collections of states at once rather than a single state at each step. There have been a few attempts \cite{Do} to extend this approach to the quantum setting, but these again rely on a logical language for representing quantum data and operations. The symbolic operator logic $\mathbf{SOL}$ proposed in this paper can serve this purpose.

\section{Conclusion}\label{sec-concl}

In this paper, we defined a two-layered logical framework in which symbolic operator logic $\mathbf{SOL}$ is building upon a first-order logic $\mathcal{L}$. This work is essentially an extension of formal and parameterised quantum states and predicates introduced in \cite{Ying24a, Ying24b}. Two logics $\mathcal{L}$ and $\mathbf{SOL}$ are designed for specification and reasoning about classical data and operations and quantum ones, respectively.

\textbf{One layer versus two layers}: One may argue that $\mathbf{SOL}$ and $\mathcal{L}$ can be compressed into the same layer to form a single logic, provided quantum data and classical data are treated in the same layer but with different types. In principle, this is possible. But in practice of implementing verification tools, we believe that the two-layer framework can make the verification of quantum computation and information in a well-structured way. 

\textbf{Symbolic specification and reasoning}: The major benefit of introducing $\mathbf{SOL}$ is that we can conveniently define quantum data and operations with classical variables in $\mathcal{L}$ as parameters, and then symbolically reason about their properties by leveraging the power of $\mathcal{L}$. This is often much more economic than non-symbolic specification and reasoning, as shown by our examples.         

{\vskip 4pt}

We must admit that this paper only presents an outline of the logical framework for symbolic specification and reasoning for quantum data and operations. Many details are still to be filled in. But we believe that the implementation of this framework in a proof assistant such as Lean and Coq will form an effective foundation for formal verification of quantum computation and information.   

{\vskip 4pt}

To conclude the paper, we point out several topics for further research in this direction:
\begin{itemize}
\item Only very simple examples are presented in this paper, and they are insufficient to demonstrate the advantages of the logic $\mathbf{SOL}$ effectively. It would be desirable to identify more compelling application examples.

\item Super-operators are extensively employed in quantum information and computation research \cite{NC00} as mathematical models of quantum channels and noise, as well as semantic models of quantum programs. However, at the current stage, they are not formalised in the logic $\mathbf{SOL}$. It is therefore important to incorporate super-operators into $\mathbf{SOL}$ by introducing a third layer of super-operators.

\item There have already been several proposals for symbolic reasoning about quantum computing in the literature. An important topic for future research is how the formal operators and logical formulas in $\mathbf{SOL}$ can be connected to other representations of quantum circuits, including various decision diagrams (DDs), such as QMDD \cite{QMDD}, TDD \cite{TDD}, CFLOBDD \cite{CFL}, LimDD \cite{Lim}, and FeynmanDD \cite{Fey}; the path-sum representation \cite{Amy18}; the ZX-calculus \cite{ZX}; tree automata \cite{Chen}; symbolic reasoning frameworks in Coq \cite{Deng} and Maude \cite{Do}; as well as unitary expressions \cite{You25}.
\end{itemize}

\bibliographystyle{ACM-Reference-Format}

\end{document}